\documentclass[12pt]{amsart}
\usepackage{amsmath,amsfonts,amssymb} \usepackage{color}
\usepackage[all,cmtip]{xy}
\usepackage{graphicx}
 \usepackage{youngtab}

\newcommand{\corr}[1]{\langle {#1} \rangle}
\newcommand{\corrr}[1]{\langle\langle {#1} \rangle\rangle}

\newcommand{\ba}{{\bf a}}

  \newcommand{\cF}{\mathcal{F}} 
 
\newcommand{\bt}{{\bf t}}
\newcommand{\bT}{{\bf T}} \newcommand{\bS}{{\bf S}}

 \newcommand{\cL}{\mathcal{L}}
  
 \newcommand{\bZ}{\mathbb{Z}}
 \newcommand{\bC}{\mathbb{C}}
 \newcommand{\pd}{\partial}
\newcommand{\Mbar}{\overline{\mathcal M}}
\newcommand{\bx}{\mathbf{x}} 
\newcommand{\vac}{|0\rangle} \newcommand{\lvac}{\langle 0|}

 \DeclareMathOperator{\Ai}{Ai} \DeclareMathOperator{\Bi}{Bi}
  \DeclareMathOperator{\res}{res}  \DeclareMathOperator{\Tr}{Tr}
 \DeclareMathOperator{\diag}{diag}   
\DeclareMathOperator{\Span}{span} 
\DeclareMathOperator{\Gr}{Gr} \DeclareMathOperator{\Aut}{Aut} \DeclareMathOperator{\sign}{sign}

\newcommand{\be}{\begin{equation}}
\newcommand{\ee}{\end{equation}}
\newcommand{\bea}{\begin{eqnarray}}
\newcommand{\eea}{\end{eqnarray}}
\newcommand{\ben}{\begin{eqnarray*}}
\newcommand{\een}{\end{eqnarray*}}

\newcommand{\half}{\frac{1}{2}}

\newtheorem{cor}{Corollary}[section]

\newtheorem{lem}[cor]{Lemma}
 \newtheorem{prop}[cor]{Proposition}
 \newtheorem{thm}[cor]{Theorem}
\theoremstyle{remark}

\definecolor{A}{rgb}{.75,1,.75}

\definecolor{green}{rgb}{0,1,0}
\definecolor{yellow}{rgb}{1,1,0}
\definecolor{orange}{rgb}{1,.7,0}
\definecolor{red}{rgb}{1,0,0}
\definecolor{white}{rgb}{1,1,1}

 \makeindex
\begin{document}
\title
{Emergent Geometry and Mirror Symmetry of A Point}
%\author{ }
%\thanks{ }

\author{Jian Zhou}
\address{Department of Mathematical Sciences\\Tsinghua University\\Beijng, 100084, China}
\email{jzhou@math.tsinghua.edu.cn}

\begin{abstract}
By considering the partition function of the topological 2D gravity,
a conformal field theory on the Airy curve emerges as the mirror theory of Gromov-Witten theory of a point.
In particular,
a formula for bosonic $n$-point functions in terms of fermionic 2-point function for this theory
is derived.
\end{abstract}

\dedicatory{Dedicated to the memory of Professor Linsheng Yin} 

\maketitle

\section{Introduction}

Witten Conjecture/Kontsevich Theory and mirror symmetry are two main themes
in Gromov-Witten theory and its many recent variants.
A combination of them lead us to study the mirror symmetry of a point in \cite{Zhou-Quant-Def}.
This involves
the Witten-Kontsevich tau-function $\tau_{WK}$ which encodes some
intersection numbers on Deligne-Mumford compactifications
of moduli spaces of algebraic curves.
As a result $\tau_{WK}$ is a formal power series in infinitely many variables $t_0, t_1, t_2, \dots$,
where $t_0$ is the coupling constant of a primary field
and it gives a linear coordinate on the small phase space,
and $t_n$ ($n\geq 1$) are the coupling constants of some descendant fields
and they give linear coordinates on the big phase space.
By emergent geometry we mean the geometric structures that naturally appear
when one considers such situations with infinite degrees of freedom,
which however cannot be seen easily when one restricts to the finite-dimensional subspaces.

Let us make an analogy between the study of Gromov-Witten like theories
and statistical physics.
In the latter one has a system consisting of large amount of particles,
and one can roughly approximate the system as having an infinite degrees of freedom.
In the former,
one has a finite-dimensional space called the small phase space
whose elements in a basis will be referred to as primary fields,
each of which has infinitely many gravitational descendents.
The primaries and their descendents generate the infinite-dimensional big phase space.
The partition function and the free energy are formal power series formally regarded as defined
on the big phase space.
So for both problems one has an infinite degree of freedom.
The genus zero part of the free energy for a Gromov-Witten like theory restricts to the small phase space
to give it a geometric structure of a Frobenius manifold \cite{Dubrovin}.
A major problem is to reconstruct from this structure the whole the theory \cite{Givental, Dubrovin-Zhang}.
Of course one does not expect to achieve this goal from the Frobenius manifold
structure alone without knowing in advance some properties satisfied
by the free energy on the big phase space.
So in the theory of Frobenius manifold,
even though one is dealing with a reconstruction problem,
some considerations from the emergent point of view is necessary,
e.g.,
one has to assume the existence of a tau-function which satisfies the loop equation  \cite{Dubrovin-Zhang}.

As advocated in an article by Anderson \cite{Anderson} and a book by Laughlin \cite{Laughlin},
there are the reduction/reconstruction and the emergence approaches to science.
In the former approach collective behavior of large quantity of individual particles
is derived from the fundamental laws obeyed by each individual particle,
in the latter approach there are fundamental laws at each level of complexity \cite{Anderson},
and it is even possible that all the fundamental laws for individual particles
have their origins in their collective behavior \cite[Preface, XV]{Laughlin}.

The analogy with statistical physics suggests that for Gromov-Witten like theories,
one might take an emergent approach to derive everything from universal properties of the free energy
in all genera on the big phase space.
One is then led to understand Frobenius manifold structure,
integrable hierarchy, W-constraints, spectral curves, etc.,
as all emerge from general properties of the free energy on the big space.
We will take topological recursion relations as the universal relations satisfied by Gromov-Witten like theories,
and show how other structures emerge naturally
as one develops the theories based on them.
We call the procedures of doing this the emergent geometry
of the corresponding theory.
In this paper we will focus on the theory of topological 2D gravity
whose partition function is Witten-Kontsevich tau-function.
Generalization to other theories related to general Frobenius manifolds will be reported
in a separate paper \cite{Zhou-Preparation}.

The function $\tau_{WK}$ has the following two equivalent characterizations \cite{Witten, Kontsevich}:
(a) It is a tau-function of the KdV hierarchy and it satisfies the puncture equation;
(b) It satisfies the Virasoro constraints.
We can now understand them from an emergent point of view.
In an earlier paper \cite{Zhou-Quant-Def}  the author proposed  a notion of a quantum deformation theory
to study mirror symmetry.
Roughly speaking, this means a deformation theory which encodes the information of free energy in all genera
on the big phase space.
For such a theory,
the moduli space is infinite-dimensional so as to encode all information of gravitational descendants in genus zero,
and furthermore,
it  admits a natural quantization from which one can produce constraints that determines the free energy in all genera.
We will recall the case of topological 2D gravity and mirror symmetry of a point
in \S \ref{sec:Qunt-Deform} below and reformulate it from the point of view of emergent geometry.

Note in \cite{Zhou-Quant-Def}
we have understood the appearance of Virasoro constraints from the quantum deformation theory of the Airy
curve and have not addressed the appearance of KdV hierarchy from this point of view.
That is exactly what we will do in this paper.
We will show that quantum deformation theory naturally leads us to a version of noncommutative deformation theory
and Kyoto school's approach to integrable hierarchies.
We will also discuss the emergence of Airy functions and its appearance in
the proof of Witten Conjecture \cite{Kontsevich, Okounkov} and
some related recent works on explicit formula for $\tau_{WK}$ \cite{Zhou-Explicit, Balogh-Yang}
and $n$-point functions \cite{Bertola-Dubrovin-Yang, Zhou-Absolute}.

As a complement of the results of \cite{Zhou-Quant-Def},
we now have a more complete formulation of the mirror symmetry of a point.
The mirror of a point is the following emergent conformal field theory living
near the infinity of the Airy curve $y = \half x^2$,
it is determined by a vector in fermionic Fock space:
\be
|W \rangle = e^A \vac \in \cF^{(0)},
\ee
where the operator $A$ of the form
$$A = \sum_{m, n \geq 0} A_{m,n} \psi_{-m-1/2} \psi_{-n-1/2}^*$$
is specified by \cite{Zhou-Explicit, Balogh-Yang},
and can be found as follows.
Let
\bea
&& a(x) = \sum_{m=0}^\infty  \frac{(6m-1)!! }{6^{2m} (2m)!} x^{-3m}, \\
&& b(y) = -\sum_{m=0}^\infty  \frac{(6m-1)!!}{6^{2m} (2m)!}
\frac{6m+1}{6m-1} y^{-3m+1}.
\eea
Then $A(x, y) = \sum_{m, n \geq 0} A_{n,m} x^{-m-1} y^{-n-1}$ is given by:
\be
A(x, y)
= \frac{a(-x) \cdot b(y) - a(y) b(-x)}{x^2 - y^2} - \frac{1}{x-y}.
\ee
Furthermore,
the $n$-point functions at $t_0=t_1 = \cdots = 0$ are given as follows.
For $F= \log Z_{WK}$, $T_{2n+1} = (2n+1)!! t_n$,
\be
\begin{split}
& \sum_{j_1,\dots, j_n \geq 1}
\frac{\pd^n F}{\pd T_{j_1} \cdots \pd T_{j_n} } \biggl|_{\bT =0}
  \xi_1^{-j_1-1}\cdots \xi_n^{-j_n-1} \\
= & (-1)^{n-1} \sum_{\text{$n$-cycles}}  \prod_{i=1}^n \hat{A}(\xi_{\sigma(i)}, \xi_{\sigma(i+1)}),
\end{split}
\ee
where
\be
\hat{A}(\xi_i, \xi_j) = \begin{cases}
A(\xi_i, \xi_j), & \text{if $i=j$}, \\
\frac{1}{\xi_i - \xi_j} + A(\xi_i, \xi_j), & \text{if $i \neq j$}.
\end{cases}
\ee
In summary,
we have the following picture for mirror symmetry of a point:
$$
\text{GW theory of a point} \leftrightarrows \text{Conformal field theory on $y=\half x^2$}
$$

The series $a(x)$ and $b(x)$ are related to the asymptotic expansions of the Airy function $\Ai(x)$
and its derivative $\Ai'(x)$ respectively.
The Airy function is a solution of the Airy equation which can be obtained by a quantization of the Airy curve.
In quantum deformation theory of the Airy curve \cite{Zhou-Quant-Def}
we consider the generalization of the miniversal deformaition
\be
y = \half x^2 + t_0,
\ee
specified by the Virasoro constraints,
in this paper we consider the deformation of its quantization
$\half \pd_{t_0}^2 + t_0$ to $\half \pd_{t_0}^2 + u(\bT)$,
specified by the KdV hierarchy.
We summarize these in the following picture:
$$
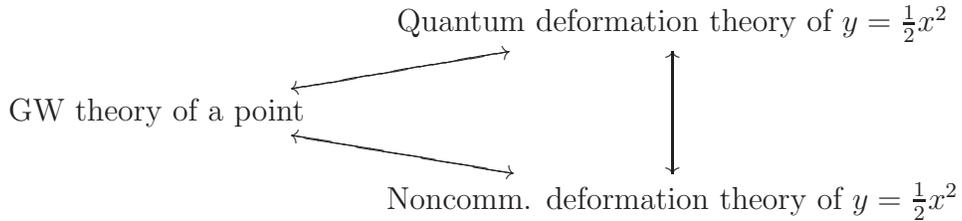
\xymatrix@R=0.5cm{
                &          \text{Quantum deformation theory of $y=\half x^2$}   \ar[dd]^{ } \ar[dl]    \\
\text{GW theory of a point} \ar[ur]^{ } \ar[dr]_{ }                 \\
                &         \text{Noncomm. deformation theory of $y = \half x^2$}     \ar[uu]\ar[ul]           }
$$
We will report on generalizations of the above pictures in a separate paper \cite{Zhou-Preparation}.

The other Sections of this paper are arranged as follows.
In Section 2 we explain the emergence of quantum deformation theory from
topological recursion relations.
In Section 3 we explain how quantum deformation theory leads naturally to Sato's Grassmannian,
KP hierarchy and a conformal field theory.
In Section 4 we use fermionic one-point functions to understand the KP hierarchy and its tau-functions,
and in Section 5 we establish a formula for bosonic $n$-point functions in terms of fermionic two-point function.
In Section 6 we explain the emergence of Airy function and its applications.

In this work we have presented an approach to KP hierarchy from a conformal field theory point of view,
and treat the applications to KdV hierarchy as a special case.
This is because we are preparing for a uniform treatment \cite{Zhou-Preparation} of all Witten $r$-spin curve 
intersection numbers,
and other cases related to semisimple Frobenius manifolds. 
See e.g. \cite{Balogh-Yang-Zhou} for applications to the case of r-spin curves. 

\vspace{.1in}

{\em Acknoledgements}.
This research is partially supported by NSFC grant 11171174.
The author thanks Professor Lu Yu for bringing his attentions to the notion of emergent phenomena
which turns out to be crucial for this work.
The author also thanks Professors Ference Balogh and Di Yang for very helpful 
communications and collaborations.

\section{Emergent Reformulation of Quantum Deformation Theory of the Airy Curve}

\label{sec:Qunt-Deform}

In this Section we first briefly recall the quantum deformation theory of the Airy curve as developed
in \cite{Zhou-Quant-Def},
then we present an emergent interpretation by topological recursion relations.

\subsection{Quantum deformation theory of the Airy curve}

We first use the Virasoro constraints to compute the genus zero one-point function on the small phase space:
\be
\frac{\pd F_0}{\pd t_n}(t_0, 0, \dots)  = \frac{1}{(n+2)!} t_0^{n+2}.
\ee
Next we note the Puiseux series:
\be
x = f - \frac{t_0}{f} - \sum_{n \geq 0} (2n+1)!! \frac{\pd F_0}{\pd t_n}(t_0, 0, \dots) \cdot f^{-2n-3},
\ee
where $f^2 = 2y$,
leads to the  miniversal deformation of the Airy curve
\be
y = \half x^2 + t_0.
\ee
When $t_0 =0$, one gets the Airy curve which is the spectral curve for
Eynard-Orantin topological recursion for topological 2D gravity \cite{Zhou-DVV}.
Next we construct a special deformation of the Airy curve of the following form:
\be \label{eqn:Special def}
x(f): = f - \sum_{n \geq 0} \frac{t_n}{(2n-1)!!} f^{2n-1}
- \sum_{n \geq 0} (2n+1)!!\frac{\pd F_0}{\pd t_n}(\bt) \cdot f^{-2n-3}.
\ee
We have proved that it is uniquely  characterized by the following property:
\be
(x(f)^2)_- = 0,
\ee
and this is equivalent to Virasoro constraints in genus zero.
To extend the picture to arbitrary genera,
the following  quantization is used.
We endow the space   of series of the form
\be
\sum_{n =0}^\infty (2n+1) \tilde{u}_n z^{(2n-1)/2}
+ \sum_{n =0}^\infty \tilde{v}_n z^{-(2n+3)/2}
\ee
 the following symplectic structure:
\be
\omega = \sum_{n =0}^\infty   d\tilde{u}_n \wedge d \tilde{v}_n.
\ee
Consider the canonical quantization:
\be
\hat{\tilde{u}}_n = \frac{t_n-\delta_{n,1}}{(2n+1)!!} \cdot, \;\;\;
\hat{\tilde{v}}_n = (2n+1)!! \frac{\pd}{\pd t_n}.
\ee
Corresponding to the field $x$,
we consider the following field of operators on the Airy curve:
\be
\hat{x}(z) = - \sum_{m \in \bZ} \beta_{-(2m+1)} z^{m-1/2}
= - \sum_{m \in \bZ} \beta_{2m+1} z^{-m-3/2}
\ee
where $f = z^{1/2}$ and the operators $\beta_{2k+1}$ are defined by:
\be
\beta_{-(2k+1)} = (2k+1) \frac{t_k-\delta_{k,1}}{(2k+1)!!} \cdot, \;\;\;\;
\beta_{2k+1} = (2k+1)!! \frac{\pd}{\pd t_k}.
\ee
We define a notion of regularized products $\hat{x}(z)^{\odot n}$
and show that the DVV Virasoro constraints  for
Witten-Kontsevich tau-function
is just the following equation:
\be
(\hat{x}(z)^{\odot 2})_- Z_{WK} = 0.
\ee

\subsection{Topological recursion relations in genus zero for topological 2D gravity}

Recall for topological 2D gravity,
the $n$-point correlators are defined by
\be
\corr{\tau_{m_1} \cdots \tau_{m_n} }_g : = \int_{\Mbar_{g,n}} \psi_1^{m_1} \cdots \psi_n^{m_n}.
\ee
The correlators in genus zero satisfy the
topological recursion relations \cite[(2.26)]{Witten}:
\be
\begin{split}
& \corr{\tau_{m_1}  \cdots \tau_{m_n}  }_0 \\
= &   \sum_{X \coprod Y = \{2, \dots, n-2\}}
\corr{\tau_{m_1-1} \prod_{j \in X} \tau_{m_j}  \cdot \tau_0 }_0
\cdot
\corr{\tau_0 \prod_{k\in Y} \tau_{m_k} \cdot \tau_{m_{n-1}} \tau_{m_n}}_0.
\end{split}
\ee
This relation reduces the calculations of the correlators in genus zero to the initial value:
\be
\corr{\tau_0^3}_0 = 1.
\ee

\subsection{Emergence of a Froebnius manifold}

In particular, one gets:
\be \label{eqn:Genus-zero}
\corr{\tau_n \tau_0^{n+2}}_0 = 1.
\ee
Of course one can also get this by the puncture equation:
\be
\corr{\tau_0 \prod_{i=1}^n \tau_{m_i}}_0
= \sum_{j=1}^n \corr{\prod_{i=1}^n \tau_{m_i-\delta_{ij}}}_0.
\ee
But the reason why we use the TRR in genus zero is that one can deduce from it
the fact that the small phase space with coordinate $t_0$ and
potential function $F_0(t_0) = \frac{t_0^3}{3!}$
is a one-dimensional Frobenius manifold.
This is the route to the emergence of the Frobenius manifold structure
that we will take when we make the generalizations in \cite{Zhou-Preparation}.

\subsection{Emergence of the Airy curve and its versal deformation}

By \eqref{eqn:Genus-zero} we have
\be
\frac{\pd F_0}{\pd t_n} (t_0) : = \frac{\pd F_0}{\pd t_n}(\bt) \biggl|_{\bt = (t_0, 0, \dots)}
= \frac{t_0^{n+2}}{(n+2)!}.
\ee
Therefore,
one can bypass the use of Virasoro constraints in \cite{Zhou-Quant-Def} to get this result.
Consider the generating series
\be \label{def:t0(z)}
t_0(z) = t_0 + \sum_{n \geq 0} \frac{\pd F_0}{\pd t_n}(t_0) \cdot z^{n+1}.
\ee
Take  a Laplace transform of $t_0(z)$:
\be \label{eqn:One-Point-Genus-Zero}
\begin{split}
\tilde{t}_0(y) & = \frac{1}{\sqrt{2\pi}} \int_0^\infty
\frac{1}{\sqrt{z}}  e^{-  y z/2} \cdot t_0(z)  dz \\
& = \frac{t_0}{y^{1/2}} + \sum_{ n\geq 0} \frac{(2n+1)!!}{y^{n+3/2}} \frac{\pd F_0}{\pd t_n}(t_0).
\end{split}
\ee
Now define a Puiseux series $x(y)$ by:
\be \label{def:p(y)}
x(y) = \tilde{t}_0(y) - y^{1/2}.
\ee
More concretely,
\be
x(y) = - y^{1/2} + \sum_{n \geq 0} \frac{(2n-1)!!}{(n+1)!} t_0^{n+1} y^{-n-1/2}
= -( y - 2 t_0)^{1/2}.
\ee
This is equivalent to:
\be \label{eqn:Airy-Def}
y = x^2 + 2t_0.
\ee
This is the versal deformation of the Airy curve:
\be \label{eqn:Airy}
y = x^2.
\ee

\subsection{Emergent interpretations by ghost variables}

Now we understand  the extra term $t_0$ on the right-hand side of \eqref{def:t0(z)},
the extra term $-y^{1/2}$ and the special deformation from an emergent point of view.
As pointed out in  \cite{Zhou-DVV},
the following convention for one-point and two-point genus zero correlators
plays a crucial role in understanding the topological 2D gravity even though they are not well defined geometrically:
\bea
&& \corr{\tau_n}_0 = \delta_{n, -2}, \\
&& \corr{\tau_k\tau_{-k-1}}_0= (-1)^k,
\eea
so that
\bea
&& \sum_{n \in \bZ} \corr{\tau_n}_0 x^n = \frac{1}{x^2}, \\
&& \sum_{k \geq 0} \corr{\tau_k\tau_{-k-1}}_0 x^k y^{-k-1} = \sum_{k \geq 0} (-1)^k x^k y^{-k-1} = \frac{1}{x+y}.
\eea
These conventions are natural since it is well-known that for $ n\geq 3$,
\be
\sum_{m_1, \dots, m_n \geq 0} \corr{\tau_{m_1} \cdots \tau_{m_n}}_0 x_1^{m_1} \cdots x_n^{m_n}
= (x_1+ \cdots +x_n)^{n-3}.
\ee
In \cite{Zhou-DVV} it was shown that these conventions determine the equation of the spectral curve
and the Bergman kernel for Eynard-Orantin recursions for topological 2D gravity.

One formally adds
\be
t_{-2} +  \sum_{n \geq 0} (-1)^n t_n t_{-n-1}
= \sum_{n \geq 0} (-1)^n (t_n - \delta_{n,1}) t_{-n-1}
\ee
to the genus zero free energy $F_0$ to get the full genus zero free energy:
\be
\tilde{F}_0 = F_0 + \sum_{n \geq 0} (-1)^n (t_n-\delta_{n,1}) t_{-n-1}.
\ee
The variables $t_{-1}, t_{-2}, \dots$ will be referred to as the ghost variables.
The space that include also these ghost variables will be called the full phase space.
Considered the augmented generating series of genus zero one-point function on the full phase space:
\be
\begin{split}
\sum_{n \in \bZ} \frac{\pd \tilde{F}_0}{\pd t_n} \cdot z^{n+1}
= & \sum_{n \geq 0} (-1)^n (t_{n}-\delta_{n-1}) z^{-n} \\
+ & \sum_{n \geq 0} \biggl( \frac{\pd F_0}{\pd t_n}(\bt)  + (-1)^n t_{-n-1} \biggr) \cdot z^{n+1}.
\end{split}
\ee
When restricted to the small phase space:
\be
\sum_{n \in \bZ} \frac{\pd \tilde{F}_0}{\pd t_n} \cdot z^{n+1} \biggl|_{t_n = \delta_{n,0} t_0}
= t_0(z) + z^{-1}.
\ee
So $t_0(z) + z^{-1}$ is the augmented genus zero one point function on the small phase space:
\be
t_0(z) + z^{-1} = \corrr{\frac{1}{z-c_1}}_0(t_0).
\ee

In the theory of Frobenius manifolds \cite{Dubrovin-Zhang},
$t_0(z)$ is the deformed flat coordinate for Dubrovin connection,
and the Laplace transform $\tilde{t}_0$ is the period of the Gauss-Manin system of the Frobenius manifold.
Such interpretation will be used in \cite{Zhou-Preparation} to make generalizations to more general
Frobenius manifolds.

Now we explain the extra term $-y^{1/2}$ on the right-hand side of \eqref{def:p(y)}.
For $a \geq 0$, one has the following formula for Laplace transform:
\be
\frac{1}{\sqrt{2\pi}} \int_0^{\infty} \frac{1}{\sqrt{z}} e^{-yz/2} \cdot z^a dz
= \frac{\Gamma(a+1/2)}{\sqrt{2\pi}} \frac{2^{a+1/2}}{y^{a+1/2}}.
\ee
We use this to define the Laplace transform of $z^{-n}$ to be
\be
\frac{\Gamma(-n+1/2)}{\sqrt{2\pi}} \frac{2^{-n+1/2}}{y^{-n+1/2}}
= \frac{(-1)^{n}}{(2n-1)!!} y^{n-1/2}.
\ee
It follows that
\be
\begin{split}
& \sum_{n \in \bZ} \frac{\pd \tilde{F}_0}{\pd t_n} \cdot \frac{1}{\sqrt{2\pi}} \int_0^{\infty} z^{n+1}
e^{-yz/2} dz \\
= & \sum_{n \geq 0} \frac{t_{n}-\delta_{n-1}}{(2n-1)!!} y^{n-1/2} \\
+ & \sum_{n \geq 0} \biggl( \frac{\pd F_0}{\pd t_n}(\bt)  + (-1)^n t_{-n-1} \biggr) \cdot  (2n+1)!! y^{-n-3/2}.
\end{split}
\ee
By restricting to the big phase space,
this provides an interpretation from emergent point of view  of the special deformation
introduced in the quantum deformation theory
of the Airy curve \cite{Zhou-Quant-Def}:
\be  \label{eqn:Special-Deform}
p(y) = \sum_{n \geq 0} \frac{t_n-\delta_{n,1}}{(2n-1)!!} y^{n-1/2}
+ \sum_{n \geq 0} (2n+1)!!\frac{\pd F_0}{\pd t_n}(\bt) \cdot y^{-n-3/2}.
\ee
A further restriction to the small phase space then yields \eqref{def:p(y)}.

\subsection{Genus zero two-point function}

Let us generalize the above treatment to genus zero two-point function on the full phase space:
\be
\begin{split}
\sum_{n_1, n_2 \in \bZ} \frac{\pd^2 \tilde{F}_0}{\pd t_{n_1} \pd t_{n_2} } \cdot z_1^{n_1+1} z_2^{n_2+1}
= & \sum_{n \geq 0} (-1)^n ( z_1^{-n} z_2^{n+1} + z_2^{-n} z_1^{n+1} )\\
+ & \sum_{n_1, n_2 \geq 0}   \frac{\pd^2 F_0}{\pd t_{n_1} \pd t_{n_2}}(\bt)  \cdot z_1^{n_1+1} z_2^{n_2+1}.
\end{split}
\ee
When restricted to the small phase space:
\be
\frac{\pd^2 F_0}{\pd t_k \pd t_l}(t_0) = \binom{k+l}{k} \frac{t_0^{k+l+1}}{(k+l+1)!},
\ee
and so
\be
\begin{split}
& \sum_{n_1, n_2 \in \bZ} \frac{\pd^2 \tilde{F}_0}{\pd t_{n_1} \pd t_{n_2} } \cdot z_1^{n_1+1} z_2^{n_2+1}\biggr|_{t_n = \delta_{n0} t_0} \\
= & \sum_{n \geq 0} (-1)^n ( z_1^{-n} z_2^{n+1} + z_2^{-n} z_1^{n+1} )\\
+ & \sum_{k, l \geq 0}   \binom{k+l}{k} \frac{t_0^{k+l+1}}{(k+l+1)!}  \cdot z_1^{k+1} z_2^{l+1}.
\end{split}
\ee
After taking the summation,
\be
\begin{split}
& \sum_{n_1, n_2 \in \bZ} \frac{\pd^2 \tilde{F}_0}{\pd t_{n_1} \pd t_{n_2} } \cdot z_1^{n_1+1} z_2^{n_2+1}\biggr|_{t_n = \delta_{n0} t_0} \\
= & i_{z_1, z_2} \frac{1}{z_1+z_2} + i_{z_2, z_1} \frac{1}{z_1+z_2} 
+ \frac{1}{z_1+z_2} (e^{(z_1+z_2)t_0}-1),
\end{split}
\ee
where
\be
\begin{split}
& i_{z_1, z_2} \frac{1}{z_1+z_2} = \sum_{n \geq 0} (-1)^n z_1^{-n} z_2^{n+1}, \\
& i_{z_2, z_1} \frac{1}{z_1+z_2} = \sum_{n \geq 0} (-1)^n z_2^{-n} z_1^{n+1}.
\end{split}
\ee
Take Laplace transform: 
\be
\begin{split}
& \sum_{n_1,n_2 \in \bZ} \frac{\pd^2 \tilde{F}_0}{\pd t_{n_1} \pd t_{n_2}} 
\frac{1}{\sqrt{2\pi}^2} \int_0^{\infty} \int_0^{\infty} z_1^{n_1+1}  z_2^{n_2+1}
e^{-y_1z_1/2-y_2z_2/2} dz_1dz_2 \\
= & \sum_{n \geq 0} (2n+1) \biggl(\frac{y_1^{n-1/2}}{y_2^{n+3/2}} + \frac{y_2^{n-1/2}}{y_1^{n+3/2}} \biggr) \\
+ & \sum_{n_1, n_2 \geq 0} \frac{\pd^2 F_0}{\pd t_{n_1} \pd t_{n_2}}(\bt)  \cdot  (2n_1+1)!! y_1^{-n_1-3/2}
\cdot (2n_2+1)!! y_2^{-n_2-3/2}.
\end{split}
\ee
By restricting to  the small phase space and taking the summations,
\be
\begin{split}
& \sum_{n_1,n_2 \in \bZ} \frac{\pd^2 \tilde{F}_0}{\pd t_{n_1} \pd t_{n_2}}(t_0)
  \int_0^{\infty} \int_0^{\infty} z_1^{n_1+1}  z_2^{n_2+1}
e^{-y_1z_1/2-y_2z_2/2} \frac{dz_1dz_2}{\sqrt{2\pi}^2} \\
= & i_{y_1, y_2} \frac{y_1+y_2}{y_1^{1/2} y_2^{1/2}(y_1-y_2)^2} 
+ i_{y_2, y_1} \frac{y_1+y_2}{y_1^{1/2} y_2^{1/2}(y_1-y_2)^2}  \\
+ & \biggl(\frac{y_1+y_2-4t_0}{(y_1-2t_0)^{1/2} (y_2-2t_0)^{1/2} (y_1-y_2)^2} 
- \frac{y_1+y_2}{y_1^{1/2} y_2^{1/2}(y_1-y_2)^2} \biggr).
\end{split}
\ee
The second line on the right-hand side is understood as a Taylor series in $t_0$.
If one takes furthermore $t_0=0$,
\be
\begin{split}
& \sum_{n_1,n_2 \in \bZ} \frac{\pd^2 \tilde{F}_0}{\pd t_{n_1} \pd t_{n_2}}\biggr|_{t_k = 0}
  \int_0^{\infty} \int_0^{\infty} z_1^{n_1+1}  z_2^{n_2+1}
e^{-y_1z_1/2-y_2z_2/2} \frac{dz_1dz_2}{\sqrt{2\pi}^2} \\
= & i_{y_1, y_2} \frac{y_1+y_2}{y_1^{1/2} y_2^{1/2}(y_1-y_2)^2}
+ i_{y_2, y_1} \frac{y_1+y_2}{y_1^{1/2} y_2^{1/2}(y_1-y_2)^2}.
\end{split}
\ee
This gives an interpretation of \cite[(26)]{Zhou-DVV} by ghost variables.

\subsection{Emergence of a conformal field theory in quantum deformation theory}

Now we consider the problem of including free energy of higher genera.
In the quantum deformation theory approach in \cite{Zhou-Quant-Def},
we consider the special deformation of the Airy curve
$y = x^2$ given by \eqref{eqn:Special-Deform}
and quantize it to the field:
\be
\hat{x}(y):=  \sum_{n \geq 0}  y^{n-1/2} \frac{t_n-\delta_{n,1}}{(2n-1)!!} \cdot
+ \sum_{n \geq 0}  y^{-n-3/2} (2n+1)!!\frac{\pd}{\pd t_n},
\ee
and define a regularized product $\hat{x}(y)^{\odot 2} = \hat{x}(y) \odot \hat{x}(y)$.
Furthermore,
the Witten-Kontsevich tau-function is shown to be uniquely determined by:
\be
\hat{x}(y)^{\odot 2} Z_{WK} = 0,
\ee
this condition is shown to be equivalent to the DVV Virasoro constraints.
A conformal field whose Fock space living near the infinity of the Airy curve emerge
in this framework.
In fact, we understand $\hat{x}(y)$ as a field of operators on the Airy curve,
with $y$ as a local coordinate near the infinity of the Airy curve.
We understand $Z_{WK}$ as an element of the 2-reduced bosonic Fock space consisting
of symmetric functions generated by the odd power functions $p_{2n+1}$'s,
associated with the infinity of the Airy curve.

\subsection{All genera one-point function on the $t_0$-line}

\label{sec:All Genera One-Point}

Define the all-genera one-point function restricted to the $t_0$-line by:
\be \label{eqn:One-Point-All-Genera}
\begin{split}
\tilde{t}_0(y; \lambda)
& =- y^{1/2} +  \frac{t_0}{y^{1/2}} + \sum_{ n\geq 0} \frac{(2n+1)!!}{y^{n+3/2}} \sum_{g \geq 0} \lambda^{2g}
\frac{\pd F_g}{\pd t_n}(t_0).
\end{split}
\ee
Note we have:
\ben
\frac{\pd F_g}{\pd t_n}(t_0)
= \sum_{m \geq 0} \corr{\tau_n \tau_0^m}_g \frac{t_0^m}{m!}
\een
where the following selection rule has to be satisfied:
\be
n = 3g-2+m
\ee
When $g=0$, $n = m-2$,
we have:
\ben
&& \sum_{m \geq 2} \frac{(2m-3)!!}{y^{m-1/2}}
\corr{\tau_{m-2}\tau_0^m}_0 \frac{t_0^m}{m!}
= \sum_{m \geq 2}  \frac{(2m-3)!!}{y^{m-1/2}} \frac{t_0^m}{m!}.
\een
For $g \geq 1$,
we have
\ben
&& \sum_{m \geq 0} \frac{(6g-3+2m)!!}{y^{3g+m-1/2}} \corr{\tau_{3g-2+m} \tau_0^m}_g \frac{t_0^m}{m!} \\
& = & \frac{1}{24^gg!} \sum_{m \geq 0}  \frac{(6g-3+2m)!!}{y^{3g+m-1/2}}\frac{t_0^m}{m!} \\
& = & \frac{(6g-3)!!}{24^gg! y^{3g-1/2}} (1-\frac{2t_0}{y})^{-3g+1/2} \\
& = & \frac{(6g-3)!!}{24^gg! (y-2t_0)^{3g-1/2}}
\een
where we have used the following fact due to Witten \cite{Witten}:
\be
\corr{\tau_{3g-2}}_g = \frac{1}{24^gg!}
\ee
and the puncture equation.
So we get:
\be \label{eqn:One-Point-All-Genera2}
p(y, t_0; \lambda)  =
\sum_{g\geq 0}  \frac{(6g-3)!!\lambda^{2g}}{24^gg! (y-2t_0)^{3g-1/2}}.
\ee
After taking derivative in $t_0$,
one gets:
\be
\pd_{t_0} p(y, t_0; \lambda)  =
\sum_{g\geq 0}  \frac{(6g-1)!!\lambda^{2g}}{24^gg! (y-2t_0)^{3g-1/2}}.
\ee
In particular
\ben
\pd_{t_0} p(y, t_0=0; \lambda)  & = &
\sum_{g\geq 0}  \frac{\lambda^{2g}}{24^gg!}  \frac{(6g-1)!!}{y^{3g+1/2}}.
\een
This series is closely related to the Airy functions.
Compare it with (2).
See also \S \ref{sec:N-Point-2DGravity}.

\subsection{Quantum corrections to the versal deformation of the Airy curve}

More explicitly,
one has:
\ben
p & = &  - (y-2t_0)^{1/2} + \frac{\lambda^2}{8} (y-2t_0)^{-5/2} + \frac{105\lambda^4}{128} (y-2t_0)^{-11/2} \\
& + & \frac{25025\lambda^6}{1024} (y-2t_0)^{-17/2} + \frac{56581525\lambda^8}{32768} (y-2t_0)^{-23/2} \\
& + & \frac{58561878375 \lambda^{10}}{262144} (y-2t_0)^{-29/2} + \cdots.
\een
By Lagrangian inversion:
\ben
(y-2t_0)^{1/2} & = & p (1+\frac{\lambda^2}{8}p^{-6}+\frac{95\lambda^4}{128}p^{-12}
+ \frac{23425}{\lambda^6}{1024} p^{-18} \\
& + & \frac{54230715\lambda^8}{32768} p^{-24}
+ \frac{56875278215\lambda^10}{262144}  p^{-30} + \cdots),
\een
and so
\ben
y -2 t_0 & = & p^2\biggl(1+ \frac{\lambda^2}{4} p^{-6}+\frac{3\lambda^4}{2} p^{-12}
+ \frac{735\lambda^6}{16} p^{-18} \\
& + & \frac{13265\lambda^8}{4} p^{-24}
+ \frac{27799785\lambda^{10} }{64} p^{-30} + \cdots \biggr).
\een
This once again indicates that one should study deformations more general than those in
traditional deformation theory.

\subsection{All genera two-point function on the $t_0$-line}

\label{sec:All Genera Two-Point}

\subsection{The $n$-point functions in genus $0$}

In \cite[Section 4.4]{Zhou-Quant-Def},
we have shown that the special deformation of the Airy curve gives rise a series
\be
y = x^2 + \cdots,
\ee
where $\cdots$ consists of terms in positive powers of $t_0, t_1, \dots$.
Let $L = y^{1/2}$ and write it as Laurent series.
In \cite{Zhou-Quant-Def}
we define the one-point function in genus zero on the big phase space by:
\be
\corrr{\phi_j }_0  =  \frac{1}{(2j+3)!!} \res (L^{2j+3} dx),
\ee
and define $n$-point function ($n \geq 2$) by taking derivatives:
\be
\corrr{\phi_{j_1}, \cdots, \phi_{j_n}}_0 = \frac{\pd}{\pd t_{j_1}} \corrr{\phi_{j_2} \cdots \phi_{j_n}}_0 .
\ee
The following result \cite[Theorem 6.4]{Zhou-Quant-Def} tells us how to recover the genus zero free energy of the 2D topological gravity
from the special deformation of the Airy curve:
\be
\corrr{\phi_{j_1}, \cdots, \phi_{j_n}}_0
= \frac{\pd^n F_0}{\pd t_{j_1} \cdots \pd t_{j_n}}.
\ee

\section{Emergence of Sato Grassmannian and Tau-Function}

Our discussion of the special deformation theory leads naturally to
Sato's Grassmannian and tau-function of KP hierarchy.
In this Section we will review this beautiful theory.
In the process we will introduce normalized basis
following \cite{Balogh-Yang, Balogh-Yang-Zhou}
and formulate Theorem \ref{thm:Bogoliubov} expressing Sato tau-function
in terms of a Bogoliubov transform.

\subsection{Sato's semi-infinite Grassmannian}

Restrict to the small phase space,
i.e., take all $t_n = 0$ except for $t_0$.
Then one gets a sequence of series of the following form:
\ben
L^{2n+1} & = & (x^2+ 2t_0)^{(2n+1)/2} = x^{2n+1} (1+\frac{2t_0}{x^2})^{(2n+1)/2} \\
& = & x^{2n+1} + (2n+1)t_0x^{2n-1}
 + \frac{(2n+1)(2n-1)}{2} t_0^2x^{2n-3}  + \cdots.
\een
Such a sequence determines a point in Sato's Grassmannian.
Let $H$ be the space   consisting of  the  formal  Laurent  series
$\sum_{n \in \bZ} a_n z^{n-1/2}$,
such that $a_n =  0$  for  $n \gg 0$.
Take $\{ z^{n-1/2} \}_{n \in \bZ}$ as a semi-infinite basis.
This means an element $\sum_{n \in \bZ} a_n z^{n-1/2} \in H$
has as coordinates a sequence of numbers $\{a_n\}_{n \in \bZ}$,
semi-infinite in the sense that $a_n = 0$ for $n \gg 0$.
This space has a natural inner product and a symplectic structure:
\bea
&& \langle f,g \rangle : = \res (f(z)g(z)dz), \\
&& \omega(f, g) : = \res (f(-z) g(z) dz).
\eea

It is clear that one has a decomposition:
\be
H = H_+ \oplus H_-,
\ee
where
$H_+ = \{\sum_{n \geq 1} a_n z^{n-1/2} \in H\}$,
$H_- = \{ \sum_{n \leq 0} a_n z^{n-1/2} \in H\}$.
Denote by $\pi_\pm: H \to H_\pm$ the natural projections.
The big cell of Sato Grassmannian $\Gr_{(0)}$ consists of linear subspaces $U \subset H$ such that
$\pi_+|_U: U \to H_+$ is an isomorphism.

\subsection{Admissible basis and Pl\"ucker coordinates}

Suppose that $U \in \Gr_{(0)}$.
Since $\{z^{n+1/2}\}_{n\geq 0}$ is a basis of $H_+$,
elements of the form
\be
f_n(z):=\pi_+^{-1}(z^{n+1/2}) = z^{n+1/2} + \sum_{k < n} c_{n, k} z^{k+1/2}
\ee
form a basis of $U$.
This shows a subspace $U$ in the big cell of Sato Grassmannian
has a basis of the form $\{f_n(z) = z^{n+1/2} + \cdots\}_{n \geq 0}$,
where $\dots$ denote lower order terms.
Conversely,
a linear subspace of $H$ with a basis of this form lies in $\Gr_{(0)}$.
Such a basis will be called an admissible basis.

Given an admissible basis
$\{f_n = z^{n+1/2} + \sum_{k < n} c_{n, k} z^{k-1/2}\}_{n\geq 0}$,
for any $N \geq 1$,
consider the expansion of the wedge product:
\ben
&& f_0 \wedge f_1 \wedge \cdots \wedge f_{N-1}.
\een
First of all we have a term of the form:
\ben
&& z^{1/2} \wedge z^{3/2} \wedge \cdots \wedge z^{N-1/2},
\een
all the other terms are of the form:
\ben
&&  z^{-m_1-1/2} \wedge  \cdots \wedge   z^{-m_{l-1} -1/2} \wedge z^{-m_l-1/2}\\
\wedge &&
z^{1/2} \wedge \cdots \wedge \widehat{z^{n_l+1/2}} \wedge \cdots
\wedge \widehat{z^{n_1+1/2}} \wedge \cdots
\wedge z^{N-1/2},
\een
for some nonnegative numbers
\be
m_1 > i_2 > \cdots > m_l \geq 0, \;\;\;\;\;
n_1 > j_2 > \cdots > n_l \geq 0.
\ee
It is straightforward to see that the coefficient of this term
can be found as follows:
Use the coefficients of $f_n$ to form the following matrix:
\ben
\begin{pmatrix}
\cdots & c_{0,-3} & c_{0,-2} & c_{0,-1} & 1 \\
\cdots & c_{1,-3} & c_{1,-2} & c_{1,-1} & c_{1,0} & 1 \\
\cdots & c_{2,-3} & c_{2,-2} & c_{2,-1} & c_{2,0} & c_{2,1} & 1 \\
  & \vdots & \vdots & \vdots & \vdots & \vdots \\
\cdots & c_{N-1,-3} & c_{N-1,-2} & c_{N-1,-1} & c_{N-1,0} & c_{N-1,1} & \cdots & 1 \\
\end{pmatrix}
\een
from this matrix one can form an $N \times N$-matrix by taking
among the columns index by negative integers those indexed
by $-m_1-1, -m_2-1, \dots, -m_l-1$,
and delete among the columns indexed by nonnegative integers that indexed by
$n_l, n_{l-1}, \dots, n_1$.
Denote by $C_{(m_1, \dots, m_l|n_1, \dots, n_l)}$ this matrix.
It is easy to see  that the determinant of this matrix is independent of $N$.
In fact,
the determinant of this matrix is equal to the determinant of a matrix
$B_{(m_1, \dots, m_l|n_1, \dots, n_l)}$
of size $(n_1+1) \times (n_1+1)$ by removing from
$C_{(m_1, \dots, m_l|n_1, \dots, n_l)}$
the columns indexed by $n > n_1$
and the corresponding rows.
For example,
$B_{(0|0)} = (c_{0,-1})$,
$B_{(1|0)} = (c_{0,-2}) $,
$B_{(0|1)} =  \begin{pmatrix} c_{0,-1} & 1 \\ c_{1,-1} & c_{1,0} \end{pmatrix} $.
For comparison,
note $B_{(100,1|2,1)}$   is of the form:
$$ \begin{pmatrix}
 c_{0,-101} & c_{0,-2} & 1 \\
 c_{1,-101} & c_{1,-2} & c_{1,0} \\
 c_{2,-101} & c_{2,-2} & c_{2,0}
\end{pmatrix},$$
it is of size $3\times 3$.
On the other hand, $B_{(2,1|100,1)}$ is of the form:
\ben
\begin{pmatrix}
c_{0,-3} & c_{0,-2} & 1 \\
c_{1,-3} & c_{1,-2} & c_{1,0} & 0 \\
c_{2,-3} & c_{2,-2} & c_{2,0} & 1 \\
c_{3,-3} & c_{3,-2} & c_{3,0} & c_{3,2} & 1 \\
\vdots & \vdots & \vdots & \vdots \\
c_{99,-3} & c_{99,-2} & c_{99,0} & c_{99,2} & \cdots & 1  \\
c_{100,-3} & c_{100,-2} & c_{100,0} & c_{100,2} & \cdots & c_{100, 99}
\end{pmatrix}
\een
It is of size $101 \times 101$.

Denote by $v_{(m_1, \dots, m_l|n_1, \dots, n_l)}$ the following expression:
\ben
&&  z^{-m_1-1/2} \wedge  \cdots \wedge   z^{-m_{l-1} -1/2} \wedge z^{-j_l-1/2}\\
\wedge &&
z^{1/2} \wedge \cdots \wedge \widehat{z^{n_l+1/2}} \wedge \cdots
\wedge \widehat{z^{n_1+1/2}} \wedge \cdots,
\een
then one has
\be \label{eqn:Det-Admissible}
 f_0 \wedge f_1 \wedge \cdots
 = \sum \det B_{(m_1, \dots, m_l|n_1, \dots, n_l)} \cdot
 v_{(m_1, \dots, m_l|n_1, \dots, n_l)}.
\ee
See   \cite{Kazarian-Lando} for a similar formula.

\subsection{Normalized basis and Pl\"ucker coordinates}

\label{sec:Normalized}

In computing the determinant of the matrix $B_{(m_1, \dots, m_l|n_1, \dots, n_l)}$,
one can use Gauss elimination to perform some row operations.
There are $n_1+1-l$ places where the entries are $1$,
so the result will be
\be
\det B_{(m_1, \dots, m_l|n_1, \dots, n_l)}
= \pm \det A_{(m_1, \dots, m_l|n_1, \dots, n_l)}
\ee
for some matrix $A_{(m_1, \dots, m_l|n_1, \dots, n_l)}$ of size $l\times l$.

This can be done more geometrically as follows.
By Gauss elimination one can see that every $U \in \Gr_{(0)}$ has an admissible basis
of the form
\be
f_n = z^{n+1/2} + \sum_{m \geq 0} a_{n,m} z^{-m - 1/2},
\ee
such a basis will be called a normalized basis.
The coefficients $\{a_{n,m}\}$ are called the affine coordinates on the big cell \cite{Balogh-Yang}.
Then one has
\ben
&& f_1 \wedge f_2 \wedge \cdots \\
& = & \sum \alpha_{m_1, \dots, m_l; n_1, \dots, n_l}\cdot z^{-m_1-1/2} \wedge \cdots \wedge z^{-m_l-1/2} \\
&& \wedge z^{1/2} \wedge \cdots \widehat{z^{n_l+1/2}} \wedge \cdots \wedge \widehat{z^{n_1+1/2}} \wedge \cdots,
\een
where $m_1> m_2 > \cdots > m_l \geq 0$, $n_1 > n_2 > \cdots > n_l \geq 0$ are two sequences of integers,
and
\be
\alpha_{m_1, \dots, m_l; n_1, \dots, n_l}
= (-1)^{n_1+ \cdots +n_l} \begin{vmatrix}
a_{n_1, m_1} & \cdots & a_{n_1, m_l} \\
\vdots & & \vdots \\
a_{n_l, m_1} & \cdots & a_{n_l, m_l}
\end{vmatrix}
\ee

\subsection{The fermionic Fock space}
It is well-known that the notation $(m_1, \dots, m_l|n_1, \dots, n_l)$
is the Frobenius notation of a partition $\mu= (\mu_1, \dots, \mu_l)$,
and numbers $m_i$ and $n_i$ are given by:
\bea
m_i = \mu_i - i, \;\;\;\;
n_i = \mu^t_i - i,
\eea
where $\mu^t$ is the conjugate partition of $\mu$.
Furthermore, one has
\ben
&&  z^{-m_1-1/2} \wedge  \cdots \wedge   z^{-m_{l-1} -1/2} \wedge z^{-j_l-1/2}\\
&& \wedge
z^{1/2} \wedge \cdots \wedge \widehat{z^{n_l+1/2}} \wedge \cdots
\wedge \widehat{z^{n_1+1/2}} \wedge \cdots \\
& = & z^{1/2-\mu_1} \wedge z^{3/2-\mu_2} \wedge \cdots.
\een

The above discussions naturally lead one to fermionic Fock space.
For a sequence $\ba = (a_1, a_2, \dots)$ of half-integers such that $a_1 < a_2 < \cdots$.
The elements in the set  $\{a_i\;|\; a_i < 0\}$
will be called bubbles of $\ba$,
the elements in the set $(\bZ_{\geq 0} + \half) - \{a_1, a_2, \dots\}$ will be called
holes of $\ba$.
It is clear that $\ba$ can have only finitely many bubbles.
We say $\ba$ is admissible it also has only finitely many holes.
For an admissible sequence $\ba$,
let
\be
|\ba\rangle : = z^{a_1} \wedge z^{a_2} \wedge \cdots \in \Lambda^{\frac{\infty}{2}}(H).
\ee
Suppose that $\ba$ has $k$ bubbles and $l$ holes,
Following Dirac,
$|\ba\rangle$ describes a state with $k$ electrons and $l$ positrons,
so one defines the charge of $\ba$ to be $l-k$.
The fermionic Fock space $\cF$ is the space of expressions of form:
\be
\sum_{\ba} c_\ba |\ba\rangle,
\ee
where the sum is taken over admissible sequences.
One has a decomposition:
\be
\cF = \sum_{n \in \bZ} \cF^{(n)},
\ee
where $\cF^{(n)}$ are generated by $|\ba\rangle$ with charge $n$.
The space $\cF^{(0)}$ is particularly interesting.
Its basis vectors correspond to partitions $\mu = (\mu_1,  \mu_2, \mu_3, \cdots)$,
where $\mu_1 \geq \mu_2 \geq \mu_3 \geq \mu_l > \mu_{l+1} = \mu_{l+2} = \cdots = 0$:
\be
|\mu\rangle: = z^{-(\mu_1-1/2)} \wedge z^{-(\mu_2-3/2)} \wedge \cdots.
\ee
When all $\mu_i=0$, the partition is called the empty partition,
the corresponding vector is called the fermionic vacuum vector and is denoted by $\vac$:
\be
\vac : = z^{1/2} \wedge z^{3/2} \wedge \cdots.
\ee
For $n \in \bZ$,
define
\be
|n\rangle := z^{n+1/2} \wedge z^{n+3/2} \wedge \cdots.
\ee
This is a vector in $\cF^{(n)}$.

\subsection{Creators and annihilators on $\mathcal{F}$}
As in the case of ordinary Grassmann algebra,
one can consider exterior products and inner products.
For $r \in \mathbb{Z}+\frac{1}{2}$,
define operator $\psi_r: \Lambda^{\frac{\infty}{2}}(H) \to \Lambda^{\frac{\infty}{2}}(H)$ by
\be
\psi_r |\ba\rangle = z^{r} \wedge |\ba\rangle,
\ee
and let $\psi_r^*: \Lambda^{\frac{\infty}{2}}(H) \to \Lambda^{\frac{\infty}{2}}(H)$ be defined by:
\be
\psi^*_r |\ba\rangle =
\begin{cases}
(-1)^{k+1} \cdot z^{a_1}\wedge\cdots\wedge \widehat{z^{a_k}}\wedge\cdots, & \text{if $a_k = -r$ for some $k$}, \\
0, &  \text{otherwise}.
\end{cases}
\ee
These operators have charge $-1$ and $1$ respectively.

The anti-commutation relations for these operators are
\begin{equation} \label{eqn:CR}
[\psi_r,\psi^*_s]_+:= \psi_r\psi^*_s + \psi^*_s\psi_r = \delta_{-r,s}id
\end{equation}
and other anti-commutation relations are zero.
It is clear that for $r > 0$,
\begin{align}
\psi_{r} \vac & = 0, & \psi_r^* \vac & = 0,
\end{align}
so the operators $\{\psi_{r}, \psi_r^*\}_{r > 0}$ are called the fermionic annihilators.
For a partition $\mu$,
let $z^{-m_1-1/2}$, $\dots$,  $z^{-m_k-1/2}$ be the positions of the electrons,
$z^{n_1+1/2}, \dots, z^{n_k+1/2}$ be the positions of the positrons,
$m_1 > m_2 > \cdots > m_k \geq 0$, $n_1 > \cdots > n_k \geq 0$,
then  $(m_1, m_2, . . ., m_k | n_1, n_2, . . ., n_k)$ is the Frobenius notation for $\mu$,
and one has
\be \label{eqn:Mu}
|\mu\rangle = (-1)^{n_1 + n_2 + \cdots + n_k} \psi_{-m_1-\frac{1}{2}} \psi_{-n_1-\frac{1}{2}}^*
\cdots \psi_{-m_k-1/2} \psi_{-n_k-1/2}^* |0\rangle,
\ee
and so the operators $\{\psi_{-r}, \psi_{-r}^*\}_{r > 0}$ are the fermionic creators.

\subsection{Elements of $\cF^{(0)}$ associated with points on big cell of Sato Grassmannian}

 Now we come to formulate the main result of this Section.

\begin{thm} \label{thm:Bogoliubov}
Suppose that $U$ is given by a normalized basis
$$\{f_n = z^{n+1/2} + \sum_{m \geq 0} a_{n,m} z^{-m - 1/2} \},$$
the one has
\be
|U\rangle = e^A \vac,
\ee
where $A: \cF^{(0)} \to \cF^{(0)}$ is a linear operator
\be
A = \sum_{m, n \geq 0} a_{n,m} \psi_{-m-1/2} \psi^*_{-n-1/2}.
\ee
\end{thm}

This is of course just a reformulation of \S \ref{sec:Normalized}.

\subsection{The boson-fermion correspondence}

In the rest of this Section,
we will recall boson-fermion correspondence and its applications to Sato tau-function.
The materials are well-known and can be found in e.g. \cite{MJD}.
They are included here to fix the notations for later Sections.
For any integer $n$, define an operator $\alpha_n$ on the fermionic Fock space $\mathcal{F}$ as follows:
\begin{equation*}
\alpha_n = \sum_{r\in \mathbb{Z} + \frac{1}{2}}:\psi_{-r}\psi^*_{r+n}:
\end{equation*}
When $n \neq 0$,
the effect of $\alpha_n$ on the vector $|\ba\rangle$ is the same as
the action of the shift operator $s_n: H \to H$ on $\Lambda^{\frac{\infty}{2}}(H)$.
When $n=0$,
$\alpha_0$ is the charge operator or fermionic number operator:
\be
\alpha_0 = \sum_{r > 0} (\psi_{-r}\psi^*_{r}- \psi^*_{-r}\psi_{r}).
\ee

They satisfy the following commutation relations:
\be \label{eqn:CR-bosonic}
[\alpha_m, \alpha_n] =m \delta_{m, -n}.
\ee
These operators can also arise in the following way.
Define the fermionic generating function:
\be
\psi(\xi) = \sum_{r \in \bZ +1/2} \psi_r \xi^{-r-1/2}, \;\;\;
\psi^*(\xi) = \sum_{r\in \bZ+1/2} \psi^*_r \xi^{-r-1/2}.
\ee
The commutation relations \eqref{eqn:CR} is equivalent to the following
operator product expansion:
\bea
&& \psi(\xi) \psi^*(\eta) = :\psi(\xi) \psi^*(\eta): + \frac{1}{\xi-\eta}, \\
&& \psi(\xi) \psi^*(\eta) = :\psi(\xi) \psi(\eta): , \\
&& \psi^*(\xi) \psi^*(\eta) = :\psi^*(\xi) \psi^*(\eta):.
\eea
Define the generating function of the operators $\alpha_n$ by
\be
\alpha(\xi): = \sum_{n \in \bZ} \alpha_n \xi^{-n-1}.
\ee
The fields of operators $\alpha(\xi)$, $\psi(\xi)$ and $\psi^*(\xi)$ are related as follows:
\be
\alpha(\xi)  = :\psi(\xi) \psi^*(\xi):.
\ee
The commutation relations \eqref{eqn:CR-bosonic} is equivalent to the following OPE:
\be
\alpha(\xi) \alpha(\eta) = :\alpha(\xi) \alpha(\eta):+ \frac{1}{(\xi-\eta)^2}.
\ee
One also has the following OPE's:
\bea
&& \alpha(\xi) \psi(\eta) = \frac{\psi(\xi)}{\xi-\eta} + :\psi(\xi)\psi^*(\xi) \psi(\eta):, \\
&& \alpha(\xi) \psi^*(\eta) = - \frac{\psi^*(\xi)}{\xi-\eta} + :\psi(\xi)\psi^*(\xi) \psi^*(\eta):.
\eea
There are equivalent to the following commutation relations:
\bea
&& [\alpha_m, \xi_r] = \psi_{m+r},  \label{comm:Alpha-Xi} \\
&& [\alpha_m, \xi_r^*] = - \psi_{m+r}^*. \label{com:Alpha-Xi*}
\eea

Consider the space of symmetric functions:
\be
\Lambda = \sum_{n \geq 0}^\infty \Lambda_n,
\ee
where $\Lambda_n$ is the space of homogeneous symmetric functions of degree $n$.
Let
$\mathcal{B} = \Lambda[w, w^{-1}]$
be the bosonic Fock space, where $w$ is a formal variable.
Then the  boson-fermion correspondence is a linear isomorphism
$\Phi: \mathcal{F} \rightarrow \mathcal{B}$ given by
\begin{equation}
|\ba \rangle \mapsto w^m \langle\underline{0}_m | e^{\sum_{n=1}^\infty \frac{p_n}{n}\alpha_n}|\ba\rangle ,\ \ |\ba\rangle\in \cF^{(m)}
\end{equation}
where $|\underline{0}_m\rangle = z^{\frac{1}{2}+m}\wedge z^{\frac{3}{2}+m}\wedge\cdots$.
Restricting to $\cF^{(0)}$,
$\Phi$ induces an isomorphism between $\cF^{(0)}$ and $\Lambda$. Explicitly, this isomorphism is given by
\cite[Theorem 9.4]{MJD}:
\begin{equation}\label{eqn:boson-fermion}
s_\mu = \lvac e^{\sum_{n=1}^\infty \frac{p_n}{n}\alpha_n} |\mu\rangle .
\end{equation}
Now combining \eqref{def:Sato},  \eqref{eqn:Determinant}, \eqref{eqn:Expansion-U}
and \eqref{eqn:boson-fermion},
one can get \eqref{eqn:Sato}.

\subsection{Action of fermionic operators on bosonic Fock space under boson-fermion correspondence}

It is very interesting to understand how the actions of the operators $\alpha_n$,
$\psi_r$ and $\psi_r^*$ originally on the fermionic Fock space
get transformed to actions on the bosonic Fock space after the boson-fermion correspondence.
The following are well-known (see e.g. \cite{MJD}):
\be
\Phi(\alpha_n |\ba\rangle)
= \begin{cases}
n \frac{\pd}{\pd p_n}\Phi(|\ba \rangle), & \text{if $n > 0$}, \\
- p_n \cdot \Phi(|\ba\rangle), & \text{if $n < 0$},
\end{cases}
\ee
Hence
\be
\Phi(\alpha(\xi) |\ba\rangle ) =
( \sum_{n \geq 1} n \frac{\pd}{\pd p_n} \xi^{-n-1}
+ \sum_{n=1}^\infty \xi^{n-1} p_n \cdot )  ( \Phi(|\ba\rangle).
\ee
For the fermionic operators,
one needs to introduce the vertex operators:
\be
\Psi(\xi) = :e^{\varphi(\xi)}:, \;\;\; \Psi^*(\xi) = e^{-\varphi(\xi):},
\ee
where the field $\varphi$ is defined by
\be
\varphi(\xi) = \sum_{n \in \bZ-\{0\}} \frac{\alpha_n}{-n} \xi^{-n} + \alpha_0 \log \xi + K,
\ee
which is an integral of the field $\alpha(\xi)$.
The integration constant $K$ is not well-defined as operator,
but it will be treated as a creator, and formally one requires the following
commutation rule:
\be
[K, \alpha_n] = \delta_{n,0}.
\ee
Then the vertex operators are given by
\bea
&& \Psi(\xi) = \exp (\sum_{n =1}^\infty \frac{p_n}{n} \xi^n)
\exp (-\sum_{n=1}^\infty \xi^{-n} \frac{\pd}{\pd p_n}) e^K \xi^{\alpha_0}, \\
&& \Psi^*(\xi) = \exp (-\sum_{n =1}^\infty \frac{p_n}{n} \xi^n)
\exp (\sum_{n=1}^\infty \xi^{-n} \frac{\pd}{\pd p_n}) e^{-K} \xi^{-\alpha_0},
\eea
where the actions of $e^K$ and $\xi^{\alpha_0}$ are defined by:
\be
(e^K f)(z, \bT) = z \cdot f(z, \bT), \;\;\;
(\xi^{\alpha_0} f) (z, \bT) =  f(\xi z, \bT).
\ee
Under the boson-fermion correspondence,
\bea
&& \Phi (\psi(\xi) |\ba \rangle = \Psi(\xi) \Phi(|\ba\rangle), \\
&& \Phi (\psi^*(\xi) |\ba \rangle = \Psi^*(\xi) \Phi(|\ba\rangle),
\eea

\subsection{Sato's construction of tau-functions}
Define a metric on $\cF$ such that $\{|\ba\rangle\;|\; \text{$\ba$ is admissible} \}$ is an orthonormal basis.
The fermionic vacuum is define by:
\be
\vac : = z^{1/2} \wedge z^{3/2} \wedge \cdots.
\ee
Following notations in physics literature the inner product of a vector $|v\rangle$ with $\vac$ will be denoted by
$\lvac v\rangle$.
One can easily see that
\be \label{eqn:Determinant}
\begin{split}
& \cdots \wedge z^{-3/2} \wedge z^{-1/2} \wedge |\ba\rangle  \\
= & \lvac \ba\rangle \cdot (\cdots \cdots \wedge z^{-3/2} \wedge z^{-1/2} \wedge z^{1/2} \wedge z^{3/2} \wedge\cdots).
\end{split}
\ee

Let $s_m: H \to H$ be the shift operator defined by:
\be \label{def:shift}
s_m ( z^{n-1/2}) = z^{m} \cdot z^{n-1/2} = z^{n  + m -1/2},
\ee
and let $\Gamma_+(\bT)$ be defined by
\be
\Gamma_+(\bT) = \exp \sum_{n =1}^\infty T_n s_n
\ee
Sato associated a tau-function $\tau_U$ to a subspace
$U$ spanned by an admissible basisof the form $\{f_n(z) = z^{n+1/2} + \sum_{j < n} a_{n,j} z^{j+1/2} \}_{n \geq 0}$ as follows:
\be \label{def:Sato}
\begin{split}
& \tau_U(\bT) \cdot (\cdots   \wedge z^{-3/2} \wedge z^{-1/2}
\wedge z^{1/2} \wedge z^{3/2} \wedge \cdots ) \\
= & \cdots   \wedge z^{-3/2} \wedge z^{-1/2}
\wedge \Gamma_+(\bT)(f_0(z)) \wedge \Gamma_+(\bT)(f_1(z)) \wedge \cdots.
\end{split}
\ee
Combining \eqref{comm:Alpha-Xi} with \eqref{def:shift}
and combining \eqref{eqn:Determinant} with \eqref{def:Sato},
one gets:
\be
\tau_U(\bT)
= \lvac e^{\sum_{n \geq 1} T_n \alpha_n} |U\rangle,
\ee
where $|U \rangle \in \cF^{(0)}$ is defined by:
\be
|U\rangle: = f_0(z) \wedge f_1(z) \wedge \cdots.
\ee
Since $\{|\mu\rangle\}$ form a basis of $\cF^{(0)}$,
there exists $c_\mu(U)$ such that
\be \label{eqn:Expansion-U}
|U\rangle = \sum_\mu c_\mu(U) \cdot |\mu\rangle.
\ee
The coefficients $c_\mu$ can be found using Theorem \ref{thm:Bogoliubov}.
One can understand $\tau_U(\bT)$ as the inner product of $|U\rangle$ with
$e^{\sum_{n \geq 1} T_n \alpha_{-n}}\vac$.
By boson-fermion correspondence \eqref{eqn:boson-fermion},
\be
e^{\sum_{n \geq 1} T_n \alpha_{-n}}\vac = \sum_\mu s_\mu(\bT) |\mu\rangle,
\ee
where $s_\mu(\bT)$ are the Schur functions defined as follows \cite{MacDonald}:
\be
s_\mu = \sum_\nu \frac{\chi^\mu_\nu}{z_\nu} p_\nu, \;\;\; p_\nu = \prod_i p_{\nu_i}, \;\;\; p_n = n T_n.
\ee
Therefore the tau-function admits an expansion:
\be \label{eqn:Sato}
\tau_U(\bT) = \sum_\mu c_\mu(U) s_\mu(\bT),
\ee
where the sum is taken over all partitions.

In the above the tau-function is constructed as a formal power series.
See Segal-Wilson \cite{Segal-Wilson} for an analytic  construction.

Sato's construction of the tau-function establishes a connection between
the theory of integrable hierarchies with conformal field theory:
\begin{center}
Sato Grassmannian $\to$ Fermionic Fock Space $\to$ Bosonic Fock Space \\
$U \in \Gr_{(0)}$ $\to$ $|U\rangle \in \cF^{(0)} $ $\to$ $\tau_U \in \Lambda$
\end{center}
He understood the theory of the integrable hierarchies as a dynamical systems
on the infinite-dimensional Grassmannian.
One should reverse the arrows in the above picture to get:
\ben
\tau_U =\lvac e^{\sum\limits_{n \geq 1} \frac{p_n}{n}\alpha_n}|U\rangle \in \Lambda
\to e^{\sum\limits_{n \geq 1} \frac{p_n}{n}\alpha_n}|U\rangle \in \cF^{(0)}
\to  e^{\sum\limits_{n \geq 1} \frac{p_n}{n} s_n}(U) \in \Gr_{(0)}.
\een

\section{Bosonic and Fermionic $N$-Point Functions}

In this Section
we will use the techniques in conformal field theory to understand
integrable systems based on  Sato's theory.

\subsection{Bosonic and fermionic correlation functions}

Given $U \in \Gr_{(0)}$,
we have seen that it determines a vector $|U\rangle \in \cF^{(0)}$.
One can define bosonic correlation functions
\be
\lvac \alpha(z_1) \cdots \alpha(z_n) |U\rangle
\ee
and fermionic correlation functions
\be
\langle n-m| \psi(z_1) \cdots \psi(z_m) \psi^*(w_1) \cdots \psi^*(w_n) |U\rangle
\ee
and their mixtures.
Next we will show that they naturally appear in the study of integrable systems.

\subsection{Hirota bilinear relations}

According to Sato \cite{Sato},
$\tau_U(\bT)$ is a tau-function of the KP hierarchy.
Let us recall how this can be shown.
First we show that
the Hirota bilinear relation  is satisfied by $ |U \rangle$ (cf. \cite[Theorem 9.3]{MJD}):
\be \label{eqn:Hirota-Fermion}
\sum_{r\in \bZ + 1/2} \psi_r^* |U\rangle \otimes \psi_{-r} |U \rangle = 0.
\ee
With Theorem \ref{thm:Bogoliubov},
we can give a more straightforward proof as follows.
One notes:
\ben
&& \psi(\xi) |U\rangle
= \sum_{m =0}^\infty (\psi_{-m-1/2} \xi^m - \sum_{n \geq 0} A_{n, m} \psi_{-n-1/2}\xi^{-m-1} ) |U\rangle, \\
&& \psi^*(\xi) |U\rangle
= \sum_{a =0}^\infty (\psi^*_{-a-1/2} \xi^a + \sum_{b \geq 0} A_{a, b} \psi^*_{-b-1/2}\xi^{-a-1} ) |U\rangle,
\een
it follows that
\ben
&& \res_{\xi = \infty} (\psi^*(\xi) |U\rangle \otimes \psi(\xi) |U\rangle)  \\
& = & - \sum_{a  \geq 0} \sum_{n \geq 0} A_{n, a} \cdot \psi^*_{-a-1/2}  |U\rangle \otimes \psi_{-n-1/2} |U\rangle \\
& + & \sum_{m \geq 0} \sum_{b \geq 0} A_{m,b} \cdot \psi^*_{-b-1/2} |U\rangle \otimes \psi_{-m-1/2} |U\rangle \\
& = & 0.
\een

\subsection{Wave-function, dual wave-function, and Hirota bilinear relations}
The fermionic 1-point functions are also called the wave-function and the dual wave-function respectively:
\bea
&& w(\bT; \xi) = \langle -1|e^{\sum\limits_{n \geq 1} \frac{p_n}{n}\alpha_n} \psi (\xi)  |U\rangle/\tau_U, \\
&& w^*(\bT; \xi) = \langle 1 | e^{\sum\limits_{n \geq 1} \frac{p_n}{n}\alpha_n} \psi^*(\xi) |U\rangle/\tau_U.
\eea
By the boson-fermion correspondence,
 \be
\begin{split}
w(\bT;z) = & \exp (\sum_{n =1} T_n z^n ) \frac{\exp ( - \sum \frac{z^{-n}}{n} \frac{\pd}{\pd T_n}) \tau(\bT)}{\tau(\bT)} \\
= & \exp (\sum_{n =1} T_n z^n ) \frac{ \tau(\bT-[1/z])}{\tau(\bT)},
\end{split}
\ee
and for the dual wave function:
\be
\begin{split}
w^*(\bT;z) = & \exp (-\sum_{n =1} T_n z^n ) \frac{\exp (\sum \frac{z^{-n}}{n} \frac{\pd}{\pd T_n}) \tau(\bT)}{\tau(\bT)} \\
= & \exp (\sum_{n =1} T_n z^n ) \frac{ \tau(\bT-[1/z])}{\tau(\bT)}.
\end{split}
\ee
These are called the Sato formulas.
They can be rewritten as follows:
\begin{align}
w(\bT;z) & = \frac{X(\bT;z) \tau_U(\bT)}{\tau_U(\bT)}, &
w^*(\bT;z) & = \frac{X^*(\bT;z) \tau_U(\bT)}{\tau_U(\bT)},
\end{align}
where the operators $X(\bT;z)$ and $X^*(\bT;z)$ defined by
\bea
&& X(\bT; z) =  \exp (\sum_{n =1} T_n z^n ) \cdot \exp ( - \sum \frac{z^{-n}}{n} \frac{\pd}{\pd T_n}), \\
&& X^*(\bT; z) =  \exp (-\sum_{n =1} T_n z^n ) \cdot \exp ( \sum \frac{z^{-n}}{n} \frac{\pd}{\pd T_n})
\eea
are also called the vertex operators.
The product of $X^*(\bT;w)$ with $X(\bT; z)$ is given by:
\be
X(\bT; z) X^*(\bT;w)  = \frac{z}{z-w} \cdot X(\bT; z, w),
\ee
where the operator $X(\bT;z,w)$ is defined by:
\be
X(\bT; z, w) = \exp (\sum_{n  \geq 1} T_n (z^n-w^n) ) \cdot
\exp ( - \sum_{n \geq 1} (\frac{z^{-n}}{n} - \frac{w^{-n}}{n}) \frac{\pd}{\pd T_n}).
\ee
By L'Hopital's rule,
\be
\lim_{z\to w} \frac{1}{z-w} (X(\bT; z,w)-1)
= \sum_{n \geq 1} n T_n w^{n-1}
+ \sum_{n \geq 1} w^{-n-1}\frac{\pd}{\pd T_n}.
\ee
On the fermionic Fock space,
this corresponds to
\be
\lim_{z \to w} (\psi(z) \psi^*(w) - \frac{1}{z-w}) = :\psi(w) \psi^*(w) = \alpha(w).
\ee

Also by the boson-fermion correspondence,
the Hirota bilinear relations \eqref{eqn:Hirota-Fermion} becomes:
\be
\res_{\xi = \infty} w(\bx, \xi) w^*(\bx', \xi) = 0.
\ee

\subsection{Dressing operator and the KP hierarchy}

Note:
\bea
&& w(\bT; \xi) = \exp \biggl(\sum_{n=1}^\infty  T_n \xi^n \biggr) \cdot
\frac{\tau(T_1-\xi^{-1}, T_2-\frac{1}{2} \xi^{-2}, \dots)}{\tau(T_1, T_2, \dots)}, \\
&& w^*(\bT; \xi) = \exp \biggl(-\sum_{n=1}^\infty  T_n \xi^n \biggr) \cdot
\frac{\tau(T_1+\xi^{-1}, T_2+\frac{1}{2} \xi^{-2}, \dots)}{\tau(T_1, T_2, \dots)}.
\eea
Write
\be
 w(\bT; \xi) = \exp \biggl(\sum_{n=1}^\infty  T_n \xi^n \biggr) \cdot
 (1 + \sum_{n=1}^\infty w_j \xi^{-n}).
\ee
The dressing operator $M$ is defined by:
\be
M:=1 + \sum_{n=1}^\infty w_j \pd_x^{-j}.
\ee
The dressing operator $M$ and the wave-function $w$ uniquely determine each other:
\be
w = M \exp \biggl(\sum_{n=1}^\infty  T_n \xi^n \biggr).
\ee
Let $L$ be the pseudo-differential operator defined by:
\be \label{eqn:Dressing}
L : = M \circ \pd_x \circ M^{-1}.
\ee
It is clear that
\be
L w = \xi \cdot w.
\ee
From the Hirota bilinear relations one can deduce that
\be
\frac{\pd}{\pd T_k} w = (L^k)_+ w.
\ee
The compatibility condition of the above two equations is
\be \label{eqn:KP}
\frac{\pd}{\pd T_k} L = [(L^k)_+, L].
\ee
From this one can also show that:
\be
\frac{\pd}{\pd t_k} M = - (L^k)_- M.
\ee
Indeed,
from \eqref{eqn:KP} and \eqref{eqn:Dressing} one immediately gets:
\be
[\frac{\pd}{\pd t_k} M \circ M^{-1} + (L^k)_-, L] = 0
\ee
Since $\frac{\pd}{\pd t_k} M \circ M^{-1} + (L^k)_-$ is a pseudodifferential operator
with coefficients differential polynomials in $a_1, a_2, \dots$,
so one can see that:
\be
\frac{\pd}{\pd t_k} M \circ M^{-1} + (L^k)_- = 0.
\ee
Here we use the following Lemma  easily proves by induction:

\begin{lem}
Suppose that $K = \sum_{n=1}^\infty b_n \pd_x^{-n}$ and $L = \pd + \sum_{n=1}^\infty a_n \pd_x^{-n}$
are pseudo-differential operators such that
\be
[K, L] = 0,
\ee
then $b_n$ are constants for all $ n\geq 1$.
\end{lem}

\subsection{From wave function to tau-function}

Let us give a proof of the following well-known result on the wave function of the KP hierarchy \cite{Adler-van Moerbeke}
from our point of view:

\begin{prop}
Suppose that $w(x; z)$ is the wave-function of KP hierarchy associated to $U \in \Gr_{(0)}$,
then
\be
U  = \Span \{ w(0; z), \pd_x w(0; z), \dots\}.
\ee
\end{prop}

\begin{proof}
Recall the wave-function is defined by:
\be
w(x; z) = \frac{\langle -1|e^{x \alpha_1} \psi (z)
e^{\sum\limits_{m,n \geq 0} A_{m,n} \psi_{-m-1/2}\psi^*_{-n-1/2} } \vac}
{\langle 0|e^{x \alpha_1}
e^{\sum\limits_{m,n \geq 0} A_{m,n} \psi_{-m-1/2}\psi^*_{-n-1/2} } \vac}.
\ee
We use Leibniz formula to compute its derivatives in $x$:
\ben
&& \pd_x^k w(x;z)|_{x=0} \\
& = & \sum_{i=0}^k \binom{k}{i} \pd_x^{k-i} \langle -1|e^{x \alpha_1} \psi (z)
e^{\sum\limits_{m,n \geq 0} A_{m,n} \psi_{-m-1/2}\psi^*_{-n-1/2} } \vac
\cdot \pd_x^i \frac{1}{\tau(x)} \biggl|_{x=0} \\
& = & \sum_{i=0}^k \binom{k}{i} \langle -1| \alpha_1^{k-i} \psi (z)
e^{\sum\limits_{m,n \geq 0} A_{m,n} \psi_{-m-1/2}\psi^*_{-n-1/2} } \vac
\cdot \pd_x^i \frac{1}{\tau(x)} \biggl|_{x=0}.
\een
It follows that
\ben
&& \Span \{ w(0; z), \pd_x w(0; z), \dots\}
=  \Span \{ \langle -1| \alpha_1^k \psi (z)
e^A \vac\}_{k \geq 0}.
\een
Note $\langle -1| \alpha_1^k \psi (z) e^A \vac$
is the inner product of $\alpha_{-1}^k \psi_{-1/2}\vac $
with $\psi(z)e^A\vac$.
We have:
\ben
\psi(z)e^A\vac
& = & \sum_{r \in \bZ+1/2} z^{-r-1/2} \psi_r \cdot
e^{\sum\limits_{m,n \geq 0} A_{m,n} \psi_{-m-1/2}\psi^*_{-n-1/2} } \vac \\
& = & \sum_{k \geq 0} z^k \psi_{-k-1/2} e^A\vac
- \sum_{m,n \geq 0} z^{-n-1} A_{m, n} \psi_{-m-1/2} e^A\vac \\
& = & \sum_{m \geq 0} (z^m-\sum_{n\geq 0}  A_{m, n} z^{-n-1}) \psi_{-m-1/2}
e^{\sum\limits_{m,n \geq 0} A_{m,n} \psi_{-m-1/2}\psi^*_{-n-1/2} } \vac.
\een
Now it is clear that:
\be
\langle -1|  \psi (z) e^A \vac
= 1 - \sum_{n=0}^\infty A_{0,n} z^{-n-1}.
\ee
Next recall:
\ben
&& \alpha_{-1} = \half \psi_{-1/2} \psi^*_{-1/2} + \sum_{n=1}^\infty( \psi_{-n-1/2} \psi^*_{n-1/2}
- \psi^*_{-n-1/2} \psi_{n-1/2}),
\een
and so one has
\ben
&& \alpha_{-1} \psi_{-1/2}\vac = \psi_{-3/2} \vac,
\een
it follows that
\ben
&&¡¡\langle -1| \alpha_1 \psi (z) e^A \vac= z - \sum_{n=0}^\infty A_{1,n} z^{-n-1}.
\een
Next from
\ben
&& \alpha_{-1}^2 \psi_{-1/2}\vac = \psi_{-5/2} \vac + \half \psi_{-1/2}\psi^*_{-1/2} \psi_{-3/2} \vac,
\een
one gets:
\ben
&& \langle -1| \alpha_1^2 \psi (z) e^A \vac
= z^2 - \sum_{n=0}^\infty A_{2,n} z^{-n-1} \\
& - &  \half A_{1,0} (1 - \sum_{n=0}^\infty A_{0,n} z^{-1-n})
+ \half  A_{0,0} (z - \sum_{n=0}^\infty   A_{1,n} z^{-1-n}).
\een
In general, from
\ben
\alpha_1^k \psi_{-1/2}\vac = \psi_{-k-1/2} \vac + \cdots,
\een
one sees that
$$ \langle -1| \alpha_1^k \psi (z) e^A \vac
= (z^k - \sum_{n=0}^\infty A_{k,n} z^{-n-1})  +\cdots,$$
where $\cdots$ stand for terms with degrees lower than $k$.
This completes the proof.
\end{proof}

\subsection{Fay identities}

In last subsection we have seen that the wave function $w(\bT;z)$ determines the tau-function $\tau(\bT)$,
hence it also determines the free energy $F(\bT)$,
and therefore,
it should also determine the $n$-point function
\be
\cF(\xi_1, \dots, \xi_n; \bT)
= \nabla(\xi_1) \cdots \nabla(\xi_n) F(\bT),
\ee
where
\be
\nabla(\xi) = \sum_{n \geq 1} \xi^{-n-1} \frac{\pd}{\pd T_n}.
\ee
Let us show how this can be achieved.
From the definition of the wave-function,
we have
\be
\frac{ \tau(\bT-[1/z])}{\tau(\bT)}
= \exp (-\sum_{n =1} T_n z^n ) \cdot w(\bT;z).
\ee
Change $\bT$ to $\bT+[1/\tilde{z}]$:
\ben
\frac{ \tau(\bT-[1/z]+[1/\tilde{z}])}{\tau(\bT+[1/\tilde{z}])}
& = & \exp (-\sum_{n =1} (T_n + \frac{1}{n \tilde{z}^n}) z^n ) \cdot w(\bT+[1/\tilde{z}];z) \\
& = & \frac{\tilde{z}}{\tilde{z}-z} \exp (-\sum_{n =1} T_n z^n ) \cdot w(\bT+[1/\tilde{z}];z).
\een
Take $\lim_{\tilde{z} \to z} \nabla_z$:
\ben
\frac{ \nabla(z) \tau(\bT)}{\tau(\bT+[1/z])}
& = & \lim_{\tilde{z} \to z} \pd_z \biggl(
\frac{\tilde{z}}{\tilde{z}-z} \exp (-\sum_{n =1} T_n z^n ) \cdot w(\bT+[1/\tilde{z}];z) \biggr) \\
& = & \exp (-\sum_{n =1} T_n z^n ) \cdot \lim_{\tilde{z} \to z} \frac{\tilde{z}}{(\tilde{z}-z)^2}  \biggl(
  w(\bT+[1/\tilde{z}];z)   \\
& - & (\tilde{z}-z) \sum_{n =1} n T_n z^n   \cdot w(\bT+[1/\tilde{z}];z)  \\
& + & (\tilde{z}-z) \cdot \pd_z w(\bT+[1/\tilde{z}];z) \biggr) \\
& = & \exp (-\sum_{n =1} T_n z^n ) \cdot z \lim_{\tilde{z} \to z} \biggl(
  \pd_{\tilde{z}}^2 w(\bT+[1/\tilde{z}];z)   \\
& - & 2 \sum_{n =1} n T_n z^n   \cdot \pd_{\tilde{z}} w(\bT+[1/\tilde{z}];z)  \\
& + & 2 \pd_{\tilde{z}}\pd_z w(\bT+[1/\tilde{z}];z) \biggr).
\een
So we obtain a formula of the form:
\ben
\nabla(z) F(\bT)
& = & w^*(\bT; z) \cdot z \lim_{\tilde{z} \to z} \biggl(
  \pd_{\tilde{z}}^2 w(\bT+[1/\tilde{z}];z)   \\
& - & 2 \sum_{n =1} n T_n z^n   \cdot \pd_{\tilde{z}} w(\bT+[1/\tilde{z}];z)
+ 2 \pd_{\tilde{z}}\pd_z w(\bT+[1/\tilde{z}];z) \biggr).
\een
To get $n$-point functions,
one can consider $\tau(\bT-[1/\xi_1]+[1/\eta_1]+\cdots -[1/\xi_n]+[1/\eta_n])$.
The result of this approach will involve both the wave-function and the dual wave function.
This leads us to consider the product
\be
w(\bT;\xi) w^*(\bT; \eta)
= e^{\sum_{n \geq 1} T_n(\xi^n - \eta^n)} \cdot \frac{\tau(\bT-[1/\xi])\tau(\bT+[1/\eta])}{\tau(\bT)^2}
\ee
and the Wronskian
$$\{w(\bT;\xi), w^*(\bT; \eta)\}
= \begin{vmatrix}
w(\bT;\xi) & w^*(\bT; \eta) \\ \pd_xw(\bT;\xi) & \pd_xw^*(\bT; \eta)
\end{vmatrix},
$$
and  their restrictions to the diagonal $\xi=\eta$.
Since
\ben
\pd_x w(\bT;z) & = & z \exp (\sum_{n =1} T_n z^n ) \frac{\tau_U(\bT-[1/z])}{\tau(\bT)} \\
& + & \exp (\sum_{n =1} T_n z^n ) \frac{\pd_x \tau_U(\bT-[1/z])}{\tau(\bT)} \\
& - &  \exp (\sum_{n =1} T_n z^n ) \frac{\tau_U(\bT-[1/z])\cdot \pd_x \tau_U(\bT)}{\tau(\bT)^2} \\
& = & \exp (\sum_{n =1} T_n z^n ) \frac{1}{\tau(\bT)^2} \\
&& \cdot \bigg(
z \tau_U(\bT) \tau_U(\bT-[1/z])- \{\tau_U(\bT), \tau_U(\bT)-[1/z])\}\biggr),
\een
and dually,
\ben
\pd_x w^*(\bT;z) & = & - z \exp (-\sum_{n =1} T_n z^n ) \frac{\tau_U(\bT+[1/z])}{\tau(\bT)} \\
& + & \exp (-\sum_{n =1} T_n z^n ) \frac{\pd_x \tau_U(\bT+[1/z])}{\tau(\bT)} \\
& - &  \exp (-\sum_{n =1} T_n z^n ) \frac{\tau_U(\bT+[1/z])\cdot \pd_x \tau_U(\bT)}{\tau(\bT)^2} \\
& = & - \exp (\sum_{n =1} T_n z^n ) \frac{1}{\tau(\bT)^2} \\
&& \cdot \bigg(
z \tau_U(\bT) \tau_U(\bT+[1/z])- \{\tau_U(\bT), \tau_U(\bT)+[1/z])\}\biggr),
\een
and so
\be \label{eqn:Wave-Wronskian}
\begin{split}
& \{w(\bT;\xi), w^*(\bT; \eta)\}
=  \exp \sum_{n \geq 1} T_n(\xi^n - \eta^n) \\
& \cdot \biggl( \frac{\{\tau(\bT-[1/\xi]), \tau(\bT+[1/\eta]) \}}{\tau(\bT)^2} \\
& - (\xi + \eta) \cdot \frac{\tau(\bT-[1/\xi])\tau(\bT+[1/\eta])}{\tau(\bT)^2} \biggr).
\end{split}
\ee
So we need to consider the product and the Wronskian of
$\tau(\bT-[1/\xi])$ and $\tau(\bT+[1/\eta])$.
They can be studied using Fay identity for $\tau$-function due to Sato,
again one has to double the number of variables first then take suitable limits:
\be \label{eqn:Fay}
\begin{split}
  & (s_0 - s_1)  (s_2 -  s_3)  \tau(\bT + [s_0] +  [s_1]  )  \tau(\bT  +  [s_2]  +  [s_3])   \\
+ & (s_0 - s_2)  (s_3 -  s_1)  \tau(\bT + [s_0] +  [s_2]  )  \tau(\bT  +  [s_3]  +  [s_1])   \\
+ & (s_0 - s_3)  (s_1 -  s_2)  \tau(\bT + [s_0] +  [s_3]  )  \tau(\bT  +  [s_1]  +  [s_2])   = 0.
\end{split}
\ee
For its derivation from the Hirota bilinear relations and its relation with Fay trisecant
identity for theta functions,
see e.g. \cite{Shiota}.
In \cite{Adler-van Moerbeke} this formula was used to derive the following formula:
\be \label{eqn:Wave-Product}
w^*(\bT; \eta) w(\bT; \xi) = \frac{1}{\xi-\eta}
\pd_x \biggl( \frac{X(\bT; \xi,\eta) \tau(\bT)}{\tau(\bT)} \biggr).
\ee
This was done as follows.
Take $\pd_{s_0}$ on both sides of \eqref{eqn:Fay} and take $s_0 = s_3 = 0$,
one can get:
\be \label{eqn:Diff-Fay1}
\begin{split}
& \{\tau(\bT + [s_1]), \tau(\bT + [s_2]) \} \\
= & (\frac{1}{s_2} - \frac{1}{s_1} ) (\tau(\bT + [s_1]) \tau(\bT + [s_2])
- \tau(\bT) \tau(\bT + [s_1]+[s_2]) ).
\end{split}
\ee
This is the differential Fay identity.
By changing $\bT$ to $\bT-[s_2]$,
one gets the following version \cite[(3.11)]{Adler-van Moerbeke}:
\be \label{eqn:Diff-Fay2}
\begin{split}
& \{\tau(\bT), \tau(\bT+[s_1] - [s_2]) \}  \\
= & (s_2^{-1} - s_1^{-1}) (\tau(\bT+[s_1]-[s_2]) \tau(\bT)
-  \tau(\bT+[s_1])\tau(\bT-[s_2]) ).
\end{split}
\ee
Using this one gets by a computation similar to that of $\pd_x w(\bT; z)$:
\ben
&   & \frac{1}{\xi-\eta} \pd_x \biggl( \frac{X(\bT; \xi,\eta) \tau(\bT)}{\tau(\bT)} \biggr) \\
& = & \frac{1}{\xi-\eta}
\pd_x \biggl(\exp (\sum_{n  \geq 1} T_n (\xi^n-\eta^n) ) \cdot
\frac{\tau(\bT-[1/\xi]+[1/\eta] )}{\tau(\bT)} \biggr) \\
& = & \frac{\exp (\sum_{n  \geq 1} T_n (\xi^n-\eta^n) )}{(\xi-\eta) \tau(\bT)^2}
\biggl((\xi-\eta)  \cdot
\tau(\bT-[1/\xi]+[1/\eta] ) \cdot \tau(\bT) \\
& - & \{\tau(\bT), \tau(\bT-[1/\xi]+[1/\eta] )\} \biggr) \\
& = &  \frac{\exp (\sum_{n  \geq 1} T_n (\xi^n-\eta^n) )}{  \tau(\bT)^2} \tau(\bT-[1/\xi]) \tau(\bT+[1/\eta]) \\
& = & w^*(\bT; \eta) w(\bT; \xi).
\een
This proves \eqref{eqn:Wave-Product}.
By \eqref{eqn:Wave-Wronskian} and \eqref{eqn:Diff-Fay1},
\be \label{eqn:Wave-Wronskian2}
\begin{split}
& \{w(\bT;\xi), w^*(\bT; \eta)\} \\
= & - (\xi + \eta) \cdot \exp \sum_{n \geq 1} T_n(\xi^n - \eta^n)   \cdot
\frac{\tau(\bT-[1/\xi]+[1/\eta])}{\tau(\bT)}  .
\end{split}
\ee
In particular,
after taking $\lim_{\eta \to \xi}$ on both sides of \eqref{eqn:Wave-Wronskian2}:
\be
 \{w(\bT;\xi), w^*(\bT; \xi)\} = -2 \xi .
\ee
In the same fashion one can show that
\be \label{eqn:Wave-Wronskian3}
\begin{split}
& \{w(\bT;\xi), w(\bT; \eta)\} \\
= & - (\xi - \eta) \cdot \exp \sum_{n \geq 1} T_n(\xi^n + \eta^n)   \cdot
\frac{\tau(\bT-[1/\xi]-[1/\eta])}{\tau(\bT)},
\end{split}
\ee
and for the dual wave-functions,
\be \label{eqn:Wave-Wronskian4}
\begin{split}
& \{w^*(\bT;\xi), w^*(\bT; \eta)\} \\
= & (\xi - \eta) \cdot \exp \sum_{n \geq 1} T_n(-\xi^n - \eta^n)   \cdot
\frac{\tau(\bT+[1/\xi]+[1/\eta])}{\tau(\bT)}  .
\end{split}
\ee

\subsection{A formula for bosonic one-point function}
Taking $\lim_{\xi \to \eta}$ on both sides of \eqref{eqn:Wave-Product}:
\be
w^*(\bT; \xi) w(\bT; \xi) =
\pd_x \biggl( \frac{(\sum\limits_{n \geq 1} n T_n \xi^{n-1}
+ \sum\limits_{n \geq 1} \xi^{-n-1}\frac{\pd}{\pd T_n}) \tau(\bT)}{\tau(\bT)} \biggr).
\ee
Write $\tau(\bT) = \exp F(\bT)$.
Then one has
\be
w^*(\bT; \xi) w(\bT; \xi) = 1 + \pd_x \sum\limits_{n \geq 1} w^{-n-1}\frac{\pd}{\pd T_n} F(\bT),
\ee
or maybe it is more appropriate to write it as
\be
w^*(\bT; \xi) w(\bT; \xi) = \pd_x( \sum\limits_{n \geq 1} n T_n w^{n-1} + \sum\limits_{n \geq 1} w^{-n-1}\frac{\pd}{\pd T_n} F(\bT) ).
\ee
One then formally has
\be
\sum\limits_{n \geq 1} n T_n w^{n-1} + \sum\limits_{n \geq 1} w^{-n-1}\frac{\pd}{\pd T_n} F(\bT)
= \pd_x^{-1} (w^*(\bT; \xi) w(\bT; \xi)).
\ee
It is actually possible to express the left-hand side in terms of wave-function and the dual wave-function
without taking the integral.
Take $\pd_\xi$ on both sides of \eqref{eqn:Wave-Wronskian2}:
\ben
&& \{\pd_\xi w(\bT;\xi), w^*(\bT; \eta)\} \\
& = & -  \exp \sum_{n \geq 1} T_n(\xi^n - \eta^n)   \cdot
\frac{\tau(\bT-[1/\xi]+[1/\eta])}{\tau(\bT)}  \\
& - & (\xi + \eta) \cdot  \exp \sum_{n \geq 1} T_n(\xi^n - \eta^n)  \\
&& \cdot \frac{1}{\tau(\bT)}
( \sum_{n \geq 1} n T_n\xi^{n-1}  + \sum_{n \geq 1} \xi^{-n-1} \frac{\pd}{\pd T_n} )
\tau(\bT-[1/\xi]+[1/\eta]).
\een
In the above we have used the following identity:
\be \label{eqn:Der-Wave}
\begin{split}
& \pd_z w(\bT;z) \\
=&
\exp (\sum_{n =1} T_n z^n ) \frac{\sum_{n =1} (n T_n z^n + z^{-n-1} \frac{\pd}{\pd T_n})\tau(\bT-[1/z])}{\tau(\bT)}.
\end{split}
\ee
And so after taking $\lim_{\eta\to\xi}$, one gets:
\be
\begin{split}
& \{\pd_\xi w(\bT;\xi), w^*(\bT; \xi)\}  \\
= &  -1  - 2\xi  \cdot
( \sum_{n \geq 1} n T_n\xi^{n-1}  + \sum_{n \geq 1} \xi^{-n-1} \frac{\pd}{\pd T_n} F(\bT)).
\end{split}
\ee
Similarly,
take $\pd_\eta$ on both sides of \eqref{eqn:Wave-Wronskian2}:
\ben
&& \{w(\bT;\xi), \pd_{\eta} w^*(\bT; \eta)\} \\
& = & -  \exp \sum_{n \geq 1} T_n(\xi^n - \eta^n)   \cdot
\frac{\tau(\bT-[1/\xi]+[1/\eta])}{\tau(\bT)}  \\
& + & (\xi + \eta) \cdot \exp \sum_{n \geq 1} T_n(\xi^n - \eta^n)  \\
&& \cdot
( \sum_{n \geq 1} n T_n\eta^{n-1}  + \sum_{n \geq 1} \eta^{-n-1} \frac{\pd}{\pd T_n} )
\frac{\tau(\bT-[1/\xi]+[1/\eta])}{\tau(\bT)},
\een
and so after taking $\lim_{\eta\to\xi}$, one gets:
\be
\begin{split}
& \{ w(\bT;\xi), \pd_\xi w^*(\bT; \xi)\}  \\
= &  -1  + 2\xi  \cdot
( \sum_{n \geq 1} n T_n\xi^{n-1}  + \sum_{n \geq 1} \xi^{-n-1} \frac{\pd}{\pd T_n} F(\bT)).
\end{split}
\ee
Reformulating the above results,
we get

\begin{thm}
For a $\tau$-function $\tau(\bT)$ of the KP hierarchy,
the following identities hold:
\bea
&& \sum_{n \geq 1} n T_n\xi^{n-1}  + \nabla(\xi) F(\bT)  \nonumber \\
& = & - \frac{1}{2\xi} ( \{\pd_\xi w(\bT;\xi), w^*(\bT; \xi)\}+1) \\
& = & \frac{1}{2\xi} (\{ w(\bT;\xi), \pd_\xi w^*(\bT; \xi)\} + 1) \\
& = & \frac{1}{4\xi} (\{ w(\bT;\xi), \pd_\xi w^*(\bT; \xi)\} -  \{\pd_\xi w(\bT;\xi), w^*(\bT; \xi)\}).
\label{eqn:One-Point3}
\eea
\end{thm}

In the case of KdV hierarchy,
we recover \cite[Theorem 1.2]{Bertola-Dubrovin-Yang} by \eqref{eqn:One-Point3}.

\subsection{Higher Fay identities and  bosonic $n$-point function of KP hierarchy
in terms of fermionic one-point functions}
\label{sec:Higher-Fay}

To obtain general formula for bosonic $n$-point functions
in terms of fermionic one-point functions,
one can repeatedly apply the loop operators of the following form to the formula for the $(n-1)$-point functions:
\be
\nabla(\xi) = \sum_{n \geq 1} \xi^{-n-1} \frac{\pd}{\pd T_n}.
\ee
In the case of KdV hierarchy,
one first sets $T_{2n} = 0$ and then applies instead the operator
$$\sum_{n\geq 1} \xi^{-2n-2}\frac{\pd}{\pd T_{2n-1}},$$
since the tau-function is independent of $T_{2n}$'s.
Then following \cite{Bertola-Dubrovin-Yang} one can derive their Theorem 1.7,
which is a formula for bosonic $n$-point function.
Let
\be
\Theta(\xi;\bT) = \half \begin{pmatrix}
-(ww^*)_x & - 2 ww^* \\ 2 w_x w^*_x & (ww^*)_x
\end{pmatrix},
\ee
then
\be
\sum_{j_1 \geq 1} \frac{\pd F}{\pd T_{j_1}}  \xi^{-2j_1-2} \\
= \half \Tr \Theta(\xi;\bT) - \sum_{j \geq 0} T_{2j+1} z^{2j+1},
\ee
and for $n \geq 2$,
\be
\begin{split}
& \sum_{j_1,\dots, j_n \geq 1}
\frac{\pd F}{\pd T_{j_1} \cdots \pd T_{j_n} }
  z_1^{-2j_1-2}\cdots z_n^{-2j_n-2} \\
= & -\frac{1}{n} \sum_{\sigma \in S_n}  \frac{\Tr(\Theta(z_{\sigma(1)}) \cdots \Theta(z_{\sigma(n)})}
{\prod_{i=1}^n (z_{\sigma(i)}^2 - z_{\sigma(i+1)}^2)}
- \delta_{n,2} \frac{z_1^2+z_2^2}{(z_1^2-z_2^2)^2}.
\end{split}
\ee
Write $R = R(\xi; \bT) = w(\bT; \xi) w^*(\bT; \xi)$,
then one can also write $\Theta$ as follows:
\be
\Theta(\xi; \bT)
= \half \begin{pmatrix}
- R_x & -2 R \\ R_{xx} - 2(\xi^2-2u) R & R_x
\end{pmatrix}.
\ee
In fact $R$ is a generating series of the Gelfand-Dickey polynomials.
See \cite{Zhou-Absolute} for a different proof of the formula of Bertola-Dubrovin-Yang \cite{Bertola-Dubrovin-Yang}
from this point of view.

The Bertola-Dubrovin-Yang formula suggests the possibility for general KP hierarchy
to express the   bosonic $n$-point function  in terms of the wave-function and the dual wave-function
and their derivatives in $x$ and $\xi_1, \dots, \xi_n$.
Since $\nabla(\xi_i)$ commute with $\pd_x$ and $\pd_{\xi_j}$ ($j \neq i$),
it reduces to computing $\nabla(\xi) w(\bT; \eta)$
and $\nabla(\xi) w^*(\bT; \eta)$.
One has
\ben
\nabla(\xi) w(\bT; \eta)
& = & \sum_{m \geq 1} \xi^{-m-1} \frac{\pd}{\pd T_m}
\biggl( \exp (\sum_{n =1} T_n \eta^m ) \frac{ \tau(\bT-[1/\eta])}{\tau(\bT)} \biggr) \\
& = & \sum_{m \geq 1} \xi^{-m-1} \eta^n \cdot \exp (\sum_{n =1} T_n \eta^n ) \frac{ \tau(\bT-[1/\eta])}{\tau(\bT)} \\
& + & \exp (\sum_{n =1} T_n \eta^n ) \frac{\nabla(\xi) \tau(\bT-[1/\eta])}{\tau(\bT)} \\
& - & \exp (\sum_{n =1} T_n \eta^n ) \frac{ \tau(\bT-[1/\eta])}{\tau(\bT)^2} \cdot \nabla(\xi) \tau(\bT) \\
& = & \biggl(\frac{1}{\xi -\eta} -  \nabla(\xi) F(\bT) \biggr) \cdot w(\bT; \eta) \\
& + & \exp (\sum_{n =1} T_n \eta^n ) \frac{\nabla(\xi) \tau(\bT-[1/\eta])}{\tau(\bT)},
\een
and similarly,
\ben
\nabla(\xi) w^*(\bT; \eta)
& = & \sum_{m \geq 1} \xi^{-m-1} \frac{\pd}{\pd T_m}
\biggl( \exp (-\sum_{n =1} T_n \eta^m ) \frac{ \tau(\bT+[1/\eta])}{\tau(\bT)} \biggr) \\
%%% & = & -\sum_{m \geq 1} \xi^{-m-1} \eta^n \cdot \exp (-\sum_{n =1} T_n \eta^n ) \frac{ \tau(\bT+[1/\eta])}{\tau(\bT)} \\
%%% & + & \exp (-\sum_{n =1} T_n \eta^n ) \frac{\nabla(\xi) \tau(\bT+[1/\eta])}{\tau(\bT)} \\
%%% & - & \exp (-\sum_{n =1} T_n \eta^n ) \frac{ \tau(\bT+[1/\eta])}{\tau(\bT)^2} \cdot \nabla(\xi) \tau(\bT) \\
& = & \biggl(- \frac{1}{\xi -\eta} -  \nabla(\xi) F(\bT) \biggr) \cdot w^*(\bT; \eta) \\
& + & \exp (-\sum_{n =1} T_n \eta^n ) \frac{\nabla(\xi) \tau(\bT+[1/\eta])}{\tau(\bT)}.
\een
Now note
\ben
&&  \frac{\nabla(\xi) \tau(\bT \pm [1/\eta])}{\tau(\bT)}
= \lim_{\tilde{\xi} \to \xi} \pd_\xi\frac{\tau(\bT-[1/\xi]+[1/\tilde{\xi}]\pm [1/\eta])}{\tau(\bT)}.
\een
In \cite{A-S-V},
the following generalization of   \eqref{eqn:Wave-Wronskian3} has been proved:
\be \label{eqn:Wave-Wronskian3n}
\begin{split}
& \{w(\bT;\xi_1),\dots,  w(\bT; \xi_k)\} \\
= & \prod_{1 \leq i< j \leq k} (\xi_j - \xi_i) \cdot \exp \sum_{n \geq 1} T_n \sum_{i=1}^k \xi_i^n   \cdot
\frac{\tau(\bT- \sum_{i=1}^k [1/\xi_i])}{\tau(\bT)},
\end{split}
\ee
We conjecture the following two identities should hold:
\be \label{eqn:Wave-Wronskian4n}
\begin{split}
& \{w^*(\bT;\eta_1), \dots, w^*(\bT; \eta_l)\} \\
= & \prod_{1 \leq i< j \leq l}  (\eta_i - \eta_j) \cdot \exp \sum_{n \geq 1} (-T_n\sum_{i=1}^l \eta^n_i)   \cdot
\frac{\tau(\bT+\sum_{i=1}^l [1/\eta_i])}{\tau(\bT)},
\end{split}
\ee

\be \label{eqn:Wave-Wronskian2n}
\begin{split}
& \{w(\bT;\xi_1),\dots,  w(\bT; \xi_k), w^*(\bT;\eta_1), \dots, w^*(\bT; \eta_l) \} \\
= & \prod_{1 \leq i< j \leq k} (\xi_j - \xi_i) \cdot
\prod_{\substack{1 \leq i \leq k \\ 1 \leq j \leq l}} (-\xi_i - \eta_j)
\cdot \prod_{1 \leq i< j \leq l}  (\eta_i - \eta_j) \\
& \cdot
\exp \sum_{n \geq 1} T_n (\sum_{i=1}^k \xi_i^n -\sum_{i=1}^l \eta^n_i) \\
&  \cdot
\frac{\tau(\bT- \sum_{i=1}^k [1/\xi_i] + \sum_{i=1}^l [1/\eta_i])}{\tau(\bT)}.
\end{split}
\ee
In particular,
we conjecture the following identities to hold:
\ben
&& \frac{\tau(\bT-[1/\xi_1]-[1/\xi_2]+[1/\eta_1])}{\tau(\bT)} \\
& = & \exp \sum_{n \geq 1} T_n (-\xi_1^n-\xi_2^n+\eta_1^n) \cdot \frac{1}{\xi_2-\xi_1}
\cdot \frac{1}{(-\xi_1-\eta_1)(-\xi_2-\eta_1)} \\
&& \cdot \{ w(\bT; \xi_1), w(\bT; \xi_2), w^*(\bT; \eta_1)\},
\een

\ben
&& \frac{\tau(\bT-[1/\xi_1]+[1/\eta_1]+[1/\eta_2])}{\tau(\bT)} \\
& = & \exp \sum_{n \geq 1} T_n (-\xi_1^n+\eta_1^n+\eta_2^n)
\cdot \frac{1}{(-\xi_1-\eta_1)(-\xi_1-\eta_2)} \cdot \frac{1}{\eta_1-\eta_2}\\
&& \cdot \{ w(\bT; \xi_1), w^*(\bT; \eta_1), w^*(\bT; \eta_2)\}.
\een
Assuming these identities,
one can get
\ben
&&  \frac{\nabla(\xi) \tau(\bT + [1/\eta])}{\tau(\bT)}
= \lim_{\tilde{\xi} \to \xi} \pd_\xi\frac{\tau(\bT-[1/\xi]+[1/\tilde{\xi}] + [1/\eta])}{\tau(\bT)} \\
& = & \exp( \sum_{n \geq 1} T_n \eta^n) \biggl(
\frac{1}{2\xi(\xi^2-\eta^2)}
 (-\sum_{n\geq 1} n T_n \xi^{n-1} + \pd_\xi )  \\
& - & \frac{3\xi+\eta}{4\xi^2(\xi+\eta)^2(\xi-\eta)}
\biggr)
\{ w(\bT; \xi), w^*(\bT; \xi), w^*(\bT; \eta)\},
\een
\ben
&&  \frac{\nabla(\xi) \tau(\bT - [1/\eta])}{\tau(\bT)}
= \lim_{\tilde{\xi} \to \xi} \pd_\xi\frac{\tau(\bT-[1/\xi]+[1/\tilde{\xi}] - [1/\eta])}{\tau(\bT)} \\
& = & - \exp( -\sum_{n \geq 1} T_n \eta^n) \biggl(
\frac{1}{2\xi(\xi^2-\eta^2)}
 (-\sum_{n\geq 1} n T_n \xi^{n-1} + \pd_\xi )  \\
& - & \frac{3\xi - \eta}{4\xi^2(\xi-\eta)^2(\xi+\eta)}
\biggr)
\{ w(\bT; \xi), w(\bT; \eta), w^*(\bT; \xi)\},
\een
With these identities,
one can express $\nabla(\xi) w(\bT; \eta)$
and $\nabla(\xi) w^*(\bT; \eta)$ in terms
of the Wronskians  of the wave-function and the dual wave-functions
and their derivatives in $\xi$.

\section{Bosonic N-Point Functions in Terms of Normalized Bases and Admissible Bases}

In last Section we have discussed the possibility of expressing
the bosonic $n$-point functions in terms of the fermionic 1-point functions for general KP hierarchy.
In this Section we will derive some formulas for bosonic $n$-point functions in terms of
a fermionic two-point function,
and the latter in term of
an admissible basis or a normalized basis.

\subsection{Bosonic $n$-point functions and connected bosonic $n$-point functions}

From the   point of view of bosonic conformal field theory,
one can consider the bosonic $n$-point functions:
\be
\corr{\alpha(\xi_1) \cdots \alpha(\xi_n)}_U:
= \lvac  e^{\sum\limits_{n \geq 1} T_n \alpha_n}
\alpha(\xi_1) \cdots \alpha(\xi_n)
|U\rangle/ \tau_U.
\ee
Since we have
\ben
\alpha(\xi_i) \alpha(\xi_j)
& = & \alpha(\xi_i)_+\alpha(\xi_j)_+ + \alpha(\xi_i)_-\alpha(\xi_j)_+ \\
& + &  \alpha(\xi_j)_-\alpha(\xi_i)_+
+ \alpha(\xi_i)_-\alpha(\xi_j)_- + i_{\xi_i, \xi_j} \frac{1}{(\xi_i-\xi_j)^2} \\
& = &  \alpha(\xi_j) \alpha(\xi_i),
\een
the $n$-point function
$\corr{\alpha(\xi_1) \cdots \alpha(\xi_n)}_U$ is symmetric with respect to
positions $\xi_1, \dots, \xi_n$,
so we will write it as $f(\xi_1, \dots, \xi_n)$.
In the above we have used the following notation:
\be
i_{x, y} \frac{1}{(x-y)^n} = \sum_{k \geq 0} \binom{-n}{k} x^{-n-k} y^k.
\ee
One can also define the connected $n$-point functions $\corr{\alpha(\xi_1) \cdots \alpha(\xi_n)}_U^c$,
and write it as $f^c(\xi_1, \dots, \xi_n)$.
These two kinds of $n$-point functions are related to each other by Mobi\"us inversion:
\bea
&& f(\xi_1, \dots,\xi_n)
= \sum_{I_1 \coprod \cdots \coprod I_k = [n]} f^c(\xi_{I_1}) \cdots f^c(\xi_{I_k}), \label{eqn:Mobius1} \\
&& f^c(\xi_1, \dots,\xi_n)
= \sum_{I_1 \coprod \cdots \coprod I_k = [n]} (-1)^{k-1} (k-1)! f(\xi_{I_1}) \cdots f(\xi_{I_k}),  \label{eqn:Mobius2}
\eea
where $[n]$ is the index set $\{1, \dots, n\}$,
by $I_1 \coprod \cdots \coprod I_k = [n]$ we mean a partition of $[n]$ into nonempty subsets
$I_1, \dots, I_k$, and $\xi_{I_i}= (\xi_j)_{j \in I_i}$.

Write $\tau_U = e^{F_U}$.
Note
\be
\alpha(\xi) = \sum_{n \geq 1} \xi^{-n-1} \frac{\pd}{\pd T_n} + \sum_{n \geq 1} n T_n \xi^{n-1}.
\ee
Then the bosonic one-point function is given by
\be \label{eqn:One-Point}
f^c(\xi_1) = \corr{\alpha(\xi)}_U = \sum_{n \geq 1} \frac{\pd F_U}{\pd T_n} \xi^{-n-1}
+ \sum_{n\geq 1}  n T_n \xi^{n-1} .
\ee
Similarly,
the bosonic two-point function is given by:
\ben
\corr{\alpha(\xi_1) \alpha(\xi_2)}_U
& = & \sum_{n_1, n_2 \geq 1}
\biggl(\frac{\pd^2 F_U}{\pd T_{n_1} \pd T_{n_2} } + \frac{\pd F_U}{\pd T_{n_1}}
\frac{\pd F_U}{\pd T_{n_2} } \biggr) \xi_1^{-n_1-1} \xi_2^{-n_2-1} \\
& + & \sum_{n=1}^\infty n \xi_1^{-n-1} \xi_2^{n-1}  \\
& + & \sum_{n_1 \geq 1} n_1 T_{n_1}\xi_1^{n_1-1} \cdot \sum_{n_2 \geq 1} \frac{\pd F_U}{\pd T_{n_2} }  \xi_2^{n_2-1} \\
& + & \sum_{n_1 \geq 1} \frac{\pd F_U}{\pd T_{n_1}} \xi_1^{n_1-1} \cdot \sum_{n_2 \geq 1}  n_2T _{n_2}  \xi_2^{n_2-1} \\
& + & \sum_{n_1 \geq 1} n_1 T_{n_1}\xi_1^{n_1-1} \cdot \sum_{n_2 \geq 1}  n_2T _{n_2}  \xi_2^{n_2-1}.
\een
It can be
rewritten as
\be
\begin{split}
\corr{\alpha(\xi_1) \alpha(\xi_2)}_U
= & \frac{1}{(\xi_1-\xi_2)^2} + \sum_{n_1, n_2 \geq 1}
\frac{\pd^2 F_U}{\pd T_{n_1} \pd T_{n_2} }
  \xi_1^{-n_1-1} \xi_2^{-n_2-1} \\
+ & \corr{ \alpha(\xi_1) }_U \cdot \corr{ \alpha(\xi_2) }_U.
\end{split}
\ee
From this we get
\be \label{eqn:Two-Point}
f^c(\xi_1, \xi_2) = \frac{1}{(\xi_1-\xi_2)^2} + \sum_{n_1, n_2 \geq 1}
\frac{\pd^2 F_U}{\pd T_{n_1} \pd T_{n_2} }
  \xi_1^{-n_1-1} \xi_2^{-n_2-1}.
\ee

\begin{prop}
For $m \geq 3$,
\be \label{eqn:m-Point}
f^c(\xi_1, \dots, \xi_m) =  \sum_{n_1,\dots, n_m \geq 1}
\frac{\pd^m F_U}{\pd T_{n_1} \cdots \pd T_{n_m} }
  \xi_1^{-n_1-1}\cdots \xi_m^{-n_m-1}.
\ee
\end{prop}

\begin{proof}
This can be proved by induction.
Denote by   $\tilde{f}^c(\xi_1, \dots, \xi_m)$
the right-hand side of \eqref{eqn:One-Point} or \eqref{eqn:Two-Point} or \eqref{eqn:m-Point},
then one can see that
\be
e^{-F_U} \alpha(\xi_{m+1})  e^{F_U} = \tilde{f}^c(\xi_1),
\ee
and
\be
\alpha(\xi_{m+1})_+ \tilde{f}^c(\xi_1, \dots, \xi_m)
= \tilde{f}^c(\xi_1, \dots, \xi_{m+1}).
\ee
Suppose that one has
$ f^c(\xi_1, \dots, \xi_k) = \tilde{f}(\xi_1, \dots, \xi_k)$ for $k \leq m$,
then one has:
\ben
f(\xi_1, \dots, \xi_{m+1}) & = & e^{-F_U} \alpha(\xi_{m+1}) (f(\xi_2, \dots, \xi_m) e^{F_U}) \\
& = & e^{-F_U} \alpha(\xi_{m+1}) ( e^{F_U} \sum_{I_1 \coprod \cdots \coprod I_k = [n]}
\tilde{f}^c(\xi_{I_1}) \cdots \tilde{f}^c(\xi_{I_k}) ) \\
& = & e^{-F_U} \alpha(\xi_{m+1})  e^{F_U} \cdot \sum_{I_1 \coprod \cdots \coprod I_k = [n]}
\tilde{f}^c(\xi_{I_1}) \cdots \tilde{f}^c(\xi_{I_k}) \\
& + & \alpha(\xi_{m+1})_+  \sum_{I_1 \coprod \cdots \coprod I_k = [n]} \tilde{f}^c(\xi_{I_1}) \cdots \tilde{f}^c(\xi_{I_k})  \\
& = & \tilde{f}^c(\xi_{m+1}) \cdot \sum_{I_1 \coprod \cdots \coprod I_k = [n]} \tilde{f}^c(\xi_{I_1}) \cdots \tilde{f}^c(\xi_{I_k}) \\
& + & \sum_{I_1 \coprod \cdots \coprod I_k = [n]}
\sum_{j=1}^k \tilde{f}^c(\xi_{I_1}) \cdots \tilde{f}^c(\xi_{m+1}, \xi_{I_j}) \cdots \tilde{f}^c(\xi_{I_k}).
\een
By induction hypothesis and \eqref{eqn:Mobius1},
$$
 f^c(\xi_1, \dots, \xi_{m+1}) = \tilde{f}(\xi_1, \dots, \xi_{m+1}).
$$
This finishes the proof.
\end{proof}

\subsection{Higher Fay identities and  bosonic $n$-point function of KP hierarchy
in terms of fermionic two-point functions}

Let us recall Okounkov's approach \cite{Okounkov} to higher Fay identities developed by
Adler-Shioda-van Moerbeke \cite{A-S-V}.
Consider
\be
\begin{split}
& \corr{\psi (\xi_1)   \cdots \psi(\xi_m) \psi^*(\eta_1)  \cdots \psi^*(\eta_n)}_U  \\
= &  \lvac
\psi (\xi_1)   \cdots \psi(\xi_m) \psi^*(\eta_1)  \cdots \psi^*(\eta_n) |U\rangle.
\end{split}
\ee
This can be computed in several different ways.
First one can combine
\be
\psi(\xi) \psi^*(\eta)
= \frac{1}{\xi-\eta} \Gamma_-(\{\xi\} - \{\eta\}) \Gamma_+(\{\eta^{-1}\} - \{\xi^{-1}\}),
\ee
with
\be
\Gamma_+(\{\bT\}) \Gamma_-(\{\bS\}) = e^{\sum_{n \geq 1} nT_n S_n}
\Gamma_-(\{\bS\}) \Gamma_+(\{\bT\})
\ee
to get
\be
\begin{split}
& \corr{\psi (\xi_1)   \cdots \psi(\xi_m) \psi^*(\eta_n) \cdots \psi^*(\eta_1) }_U  \\
= &  \frac{\Delta_n(\xi) \Delta_n(\eta)}{\prod_{1 \leq i, j \leq  n} (\xi_i - \eta_j)}
\tau_U ( \sum_{i=1}^n (\{\eta_i^{-1} \} - \{\xi_i^{-1}\}) ) .
\end{split}
\ee
In particular,
when $n=1$,
\be \label{eqn:Two-Point-Tau}
\corr{\psi (\xi) \psi^*(\eta) }_U  =  \frac{1}{\xi - \eta} \tau_U( \{\eta^{-1} \} - \{\xi^{-1}\} ).
\ee
On the other hand,
by Wick's Theorem,
\be \label{eqn:Wick}
\corr{\psi (\xi_1) \cdots \psi(\xi_n) \psi^*(\eta_n) \cdots \psi^*(\eta_1)  }_U
= \det(\corr{\psi (\xi_i) \psi^*(\eta_j)}_U )_{1 \leq i, j \leq n}.
\ee
Combining the preceding three identities:
\be
\begin{split}
& \det \biggl( \frac{1}{\xi_i - \eta_j} \tau_U( \{\eta_j^{-1} \} - \{\xi_i^{-1}\} ) \biggr) \\
= & \frac{\Delta_n(\xi) \Delta_n(\eta)}{\prod_{1 \leq i, j \leq  n} (\xi_i - \eta_j)}
\tau_U ( \sum_{i=1}^n (\{\eta_i^{-1} \} - \{\xi_i^{-1}\}).
\end{split}
\ee
This is a special case of the following higher Fay identities developed by Adler-Shioda-van Moerbeke
\cite[(45)]{A-S-V}:
\be
\begin{split}
& \det \biggl( \frac{1}{\xi_i - \eta_j} \frac{\tau_U(\bT+  \{\eta_j^{-1} \} - \{\xi_i^{-1}\} )}{\tau_U(\bT)} \biggr) \\
= & \frac{\Delta_n(\xi) \Delta_n(\eta)}{\prod_{1 \leq i, j \leq  n} (\xi_i - \eta_j)}
\frac{ \tau_U ( \sum_{i=1}^n (\{\eta_i^{-1} \} - \{\xi_i^{-1}\}) }{\tau_U(\bT)}.
\end{split}
\ee

\subsection{Fermionic two-point function in terms of affine coordinates}

Recall $U = e^A \vac$,
where
\be
A = \sum_{m, n \geq 0} a_{n,m} \psi_{-m-1/2} \psi^*_{-n-1/2}.
\ee
It is easy to see that
\be
\corr{\psi (\xi) \psi^*(\eta)}_U
= i_{\xi, \eta} \frac{1}{\xi-\eta} + \sum_{m,n\geq 0} a_{m,n} \xi^{-m-1} \eta^{-n-1}.
\ee
In particular,
$\corr{\psi (\xi) \psi^*(\eta)}_U$ contains the same information as the operator $A$.
Write
\be
A(\xi, \eta) =  \sum_{m,n\geq 0} a_{m,n} \xi^{-m-1} \eta^{-n-1}.
\ee
Then we have
\be
\corr{\psi (\xi) \psi^*(\eta)}_U = i_{\xi, \eta} \frac{1}{\xi-\eta} + A(\xi, \eta).
\ee

\subsection{Bosonic $n$-point function
in terms of fermionic two-point functions}
Use the OPE:
\be \label{eqn:Fermionic-OPE}
\psi (\xi) \psi^*(\eta)
= i_{\xi, \eta} \frac{1}{\xi-\eta} + :\psi(\xi)\psi^*(\eta):
\ee
one gets:
\be \label{eqn:A-Expectation}
\corr{:\psi (\xi) \psi^*(\eta):}_U =  A(\xi, \eta).
\ee
After taking $\lim_{\eta \to \xi}$:
\be
\corr{\alpha (\xi)}_U =  A(\xi, \xi).
\ee
Let us now compute $\corr{\alpha(\xi_1) \alpha(\xi_2) }_U$ in the same fashion.
On the one hand we have:
\ben
&& \corr{\psi (\xi_1) \psi^*(\eta_1) \psi(\xi_2) \psi^*(\eta_2)}_U \\
& = & - \corr{\psi (\xi_1)  \psi(\xi_2) \psi^*(\eta_1) \psi^*(\eta_2)}_U  \\
& + &  \biggl(i_{\xi_2, \eta_1} \frac{1}{\xi_2-\eta_1}
- i_{\eta_1, \xi_2} \frac{1}{\xi_2- \eta_1} \biggr) \cdot  \corr{\psi (\xi_1) \psi^*(\eta_2)}_U \\
& = & \begin{vmatrix}
\corr{\psi(\xi_1) \psi^*(\eta_1)}_U & \corr{\psi (\xi_1) \psi^*(\eta_2)}_U \\
\corr{\psi(\xi_2) \psi^*(\eta_1)}_U & \corr{\psi(\xi_2) \psi^*(\eta_2)}_U
\end{vmatrix} \\
& + & \biggl(i_{\xi_2, \eta_1} \frac{1}{\xi_2-\eta_1}
- i_{\eta_1, \xi_2} \frac{1}{\xi_2- \eta_1} \biggr) \cdot  (i_{\xi_1, \eta_2}  \frac{1}{\xi_1-\eta_2} + A(\xi_1, \eta_2) ) \\
& = & \begin{vmatrix}
i_{\xi_1, \eta_1} \frac{1}{\xi_1 - \eta_1} + A(\xi_1, \eta_1) &
i_{\xi_1, \eta_2} \frac{1}{\xi_1 - \eta_2} + A(\xi_1, \eta_2) \\
i_{\xi_2, \eta_1} \frac{1}{\xi_2 - \eta_1} + A(\xi_2, \eta_1) &
i_{\xi_2, \eta_2} \frac{1}{\xi_2 - \eta_2} + A(\xi_2, \eta_2)
\end{vmatrix} \\
& + & \biggl(i_{\xi_2, \eta_1} \frac{1}{\xi_2-\eta_1}
- i_{\eta_1, \xi_2} \frac{1}{\xi_2- \eta_1} \biggr)  \cdot
(i_{\xi_1, \eta_2} \frac{1}{\xi_1-\eta_2} + A(\xi_1, \eta_2) ) \\
& = & \begin{vmatrix}
i_{\xi_1, \eta_1} \frac{1}{\xi_1 - \eta_1} + A(\xi_1, \eta_1) &
i_{\xi_1, \eta_2} \frac{1}{\xi_1 - \eta_2} + A(\xi_1, \eta_2) \\
i_{\eta_1, \xi_2} \frac{1}{\xi_2 - \eta_1} + A(\xi_2, \eta_1) &
i_{\xi_2, \eta_2} \frac{1}{\xi_2 - \eta_2} + A(\xi_2, \eta_2)
\end{vmatrix}.
\een
In the last step we use the following identity for determinants of $2\times 2$-matrices:
\be \label{eqn:Sum-Determinant}
\begin{vmatrix}
c_{11} & c_{12} \\c_{21} & c_{22}
\end{vmatrix} + \begin{vmatrix}
0 & c_{12} \\ b_{21} & c_{22}
\end{vmatrix}
= \begin{vmatrix}
c_{11} & c_{12} \\ b_{12} + c_{21} & c_{22}
\end{vmatrix}.
\ee
On the other hand,
\ben
&& \corr{\psi (\xi_1) \psi^*(\eta_1) \psi(\xi_2) \psi^*(\eta_2)}_U \\
& = & \corr{\biggl(i_{\xi_1, \eta_1} \frac{1}{\xi_1-\eta_1} + :\psi(\xi_1) \psi^*(\eta_1):\biggr)
\biggl(i_{\xi_2, \eta_2} \frac{1}{\xi_2-\eta_2} + :\psi(\xi_2)\psi^*(\eta_2):\biggr) }_U \\
& = & i_{\xi_1, \eta_1}\frac{1}{\xi_1-\eta_1}
\cdot i_{\xi_2, \eta_2} \frac{1}{\xi_2- \eta_2}
+ i_{\xi_1, \eta_1}\frac{1}{\xi_1-\eta_1}\cdot A(\xi_2, \eta_2) \\
& + & i_{\xi_2, \eta_2}\frac{1}{\xi_2-\eta_2}\cdot A(\xi_1, \eta_1)
+ \corr{ :\psi(\xi_1) \psi^*(\eta_1):  :\psi(\xi_2)\psi^*(\eta_2):  }_U .
\een
Combining the two calculations,
one gets:
\ben
&& \corr{ :\psi(\xi_1) \psi^*(\eta_1):  :\psi(\xi_2)\psi^*(\eta_2):  }_U \\
& = & \begin{vmatrix}
 A(\xi_1, \eta_1) &
i_{\xi_1, \eta_2} \frac{1}{\xi_1 - \eta_2} + A(\xi_1, \eta_2) \\
i_{\eta_1, \xi_2} \frac{1}{\xi_2 - \eta_1} + A(\xi_2, \eta_1) &
 A(\xi_2, \eta_2)
\end{vmatrix}.
\een
Now we take $\lim_{\eta_1 \to \xi_1} \lim_{\eta_2 \to \xi_2}$:
\ben
\corr{\alpha(\xi_1) \alpha(\xi_2) }_U
= \begin{vmatrix}
A(\xi_1, \xi_1) & i_{\xi_1, \xi_2} \frac{1}{\xi_1 - \xi_2} + A(\xi_1, \xi_2) \\
i_{\xi_1, \xi_2} \frac{1}{\xi_2 - \xi_1} + A(\xi_2, \xi_1) & A(\xi_2, \xi_2)
\end{vmatrix}.
\een

Now we generalize the computation for $\corr{\psi (\xi_1) \psi^*(\eta_1) \psi(\xi_2) \psi^*(\eta_2)}_U$
to the case of  $\corr{\psi (\xi_1) \psi^*(\eta_1) \cdots \psi(\xi_n) \psi^*(\eta_n)}_U$ for $n > 2$.
One first moves all $\psi^*(\eta_j)$'s to the right of all $\psi(\xi_i)$'s,
and then apply \eqref{eqn:Wick}.
By an analogue of  \eqref{eqn:Sum-Determinant},
one gets:
\be
\corr{\psi (\xi_1) \psi^*(\eta_1) \cdots \psi(\xi_n) \psi^*(\eta_n)}_U
= \det ( C_{ij} )_{1 \leq i, j \leq n},
\ee
where
\be
C_{ij} = \begin{cases}
i_{\xi_i, \eta_j} \frac{1}{\xi_i-\eta_j} + A(\xi_i, \eta_j),  & i \leq j, \\
i_{\eta_j, \xi_i} \frac{1}{\xi_i-\eta_j} + A(\xi_i, \eta_j),  & i > j.
\end{cases}
\ee
On the other hand,
\ben
&& \corr{\psi (\xi_1) \psi^*(\eta_1) \cdots \psi(\xi_n) \psi^*(\eta_n)}_U \\
& = & \corr{( i_{\xi_1, \eta_1}\frac{1}{\xi_1-\eta_1} + :\psi(\xi_1)\psi^*(\eta_1):)  \cdots
(i_{\xi_n, \eta_n}\frac{1}{\xi_n-\eta_n} + :\psi(\xi_n)\psi^*(\eta_n):)}_U.
\een
So we get an identity of the form
\ben
&& \corr{(i_{\xi_1, \eta_1}\frac{1}{\xi_1-\eta_1} + :\psi(\xi_1)\psi^*(\eta_1):)  \cdots
(i_{\xi_n, \eta_n} \frac{1}{\xi_n-\eta_n} + :\psi(\xi_n)\psi^*(\eta_n):)}_U \\
& = &
\det(C_{ij} )_{1 \leq i, j \leq n}.
\een
It is clear that the left-hand side can be expanded into $2^n$ terms,
among which only one term does not contain a factor of the form $i_{\xi_i, \eta_i}\frac{1}{\xi_i - \eta_i}$,
that is
$\corr{:\psi(\xi_1)\psi^*(\eta_1):  \cdots :\psi(\xi_n)\psi^*(\eta_n):}_U$.
Similarly,
the right-hand can also be expanded into $2^n$ term,
this can be done by decomposing each row vector of the matrix  into a sum of two row vectors,
one of them is given by $(\delta_{i,j}i_{\xi_i, \eta_i} \frac{1}{\xi_i-\eta_i})_{j=1, \dots,n}$.
And so each of  the $2^n$ terms is a determinant,
and there is only one term that does not contain a factor of the form $i_{\xi_i, \eta_i} \frac{1}{\xi_i-\eta_i}$,
that is given by
$\det (\hat{C}_{i,j}(\xi_i, \eta_j))_{1 \leq i, j \leq n}$,
where
\be
\hat{C}_{ij} = \begin{cases}
A(\xi_i, \xi_i), & i =j, \\
C_{i,j}, & i \neq j.
\end{cases}
\ee
Therefore,
\be
\corr{:\psi(\xi_1)\psi^*(\eta_1):  \cdots  :\psi(\xi_n)\psi^*(\eta_n):}_U
=  \det(\hat{C}_{i,j} )_{1 \leq i, j \leq n}.
\ee
After taking $\lim_{\eta_i \to \xi_i}$,
we obtain:
\be
  \corr{ \alpha(\xi_1)  \cdots  \alpha(\xi_n) }_U
= \det( \hat{A}(\xi_i,\xi_j) )_{1 \leq i, j \leq n},
\ee
where
\be
\hat{A}(\xi_i, \xi_j) = \begin{cases}
i_{\xi_i, \xi_j} \frac{1}{\xi_i-\xi_j} + A(\xi_i, \xi_j),  & i < j, \\
A(\xi_i, \xi_i),  & i =j, \\
i_{\xi_j, \xi_i} \frac{1}{\xi_i-\xi_j} + A(\xi_i, \xi_j),  & i > j.
\end{cases}
\ee

\subsection{Connected bosonic $n$-point function
in terms of fermionic two-point functions}

We now prove a combinatorial result.

\begin{prop}
Suppose that we have two sequences of functions $\{\varphi(\xi_1, \dots, \xi_n)\}_{n \geq 1}$
and $\{\varphi^c(\xi_1, \dots, \xi_n)\}_{n \geq 1}$,
related to each other by Mobi\"us inversion:
\bea
&& \varphi(\xi_1, \dots,\xi_n)
= \sum_{I_1 \coprod \cdots \coprod I_k = [n]} \varphi^c(\xi_{I_1}) \cdots \varphi^c(\xi_{I_k}),  \label{eqn:Phi-Mobius}\\
&& \varphi^c(\xi_1, \dots,\xi_n)
= \sum_{I_1 \coprod \cdots \coprod I_k = [n]} (-1)^{k-1} (k-1)! \varphi(\xi_{I_1}) \cdots \varphi(\xi_{I_k}),
\eea
where the summations are taken over partitions of $[n]$ into nonempty subsets
$I_1, \dots, I_k$, and $\xi_{I_i}= (\xi_j)_{j \in I_i}$.
Suppose that that there are functions $B(\xi, \eta)$ such that
\be \label{eqn:Det-Phi-B}
\varphi(\xi_1, \dots,\xi_n)
= \det (B(\xi_i, \xi_j))_{1 \leq i, j \leq n}.
\ee
Then one has
\be \label{eqn:PhiC-Cycle}
\varphi^c(\xi_1, \dots,\xi_n)
=  (-1)^{n-1} \sum_{\text{$n$-cycles}}  \prod_{i=1}^n B(\xi_{\sigma(i)} \xi_{\sigma(i+1)}),
\ee
where the summations are taken over $n$-cycles $\sigma$, and $\sigma(n+1) = \sigma(1)$.
\end{prop}

\begin{proof}
This can be proved by induction.
When $n=1$,
it holds automatically.
Suppose that it holds for $1, \dots, n-1$.
By \eqref{eqn:Det-Phi-B},
\be
\varphi(\xi_1, \dots,\xi_n)
= \sum_{\sigma \in S_n} \sign(\sigma) \prod_{i=1}^n B(\xi_i, \xi_{\sigma(i)}).
\ee
Now we rewrite the right-hand side by writing a permutation $\sigma \in S_n$
as a product of the disjoint cycles.
If $\sigma$ is an $n$-cycle, then $\sign(\sigma) = (-1)^{n-1}$,
and so the contribution of all $n$-cycles is exactly the right-hand side of
\eqref{eqn:PhiC-Cycle},
which we denote by $\psi(\xi_1, \dots, \xi_n)$.
Therefore,
one gets:
\be
\varphi(\xi_1, \dots,\xi_n)
= \sum_{I_1 \coprod \cdots \coprod I_k = [n]} \psi(\xi_{I_1}) \cdots \psi(\xi_{I_k}).
\ee
The proof is completed by using \eqref{eqn:Phi-Mobius} and the induction hypothesis.
\end{proof}

As a corollary,

\begin{thm} \label{thm:Bosonic-N-Point}
For $n \geq 2$,
\be
\begin{split}
& \sum_{j_1,\dots, j_n \geq 1}
\frac{\pd^n F_U}{\pd T_{j_1} \cdots \pd T_{j_n} } \biggl|_{\bT =0}
  \xi_1^{-j_1-1}\cdots \xi_n^{-j_n-1} \\
= & (-1)^{n-1} \sum_{\text{$n$-cycles}}  \prod_{i=1}^n \hat{A}(\xi_{\sigma(i)}, \xi_{\sigma(i+1)}).
\end{split}
\ee
\end{thm}

\subsection{Fermionic two-point function in terms of admissible basis}

By combining \eqref{eqn:Two-Point-Tau} with \eqref{eqn:Fermionic-OPE} and \eqref{eqn:A-Expectation},
one gets:
\be
A(\xi, \eta) =  \frac{1}{\xi - \eta} \tau_U( \{\eta^{-1} \} - \{\xi^{-1}\} )
 - i_{\xi, \eta} \frac{1}{\xi-\eta}.
\ee
Let us now compute $\tau_U( \{\eta^{-1} \} - \{\xi^{-1}\} )$ in terms of admissible basis.
We first note that
\be \label{eqn:Specialization-Tau}
\tau_U( \{\eta^{-1} \} - \{\xi^{-1}\} )
= 1 + (\xi - \eta) A(\xi, \eta),
\ee
and so $\tau_U( \{\eta^{-1} \} - \{\xi^{-1}\} ) $ is of order at most one in $A_{m,n}$'s.
Recall that
\ben
\tau_U( \{\eta^{-1} \} - \{\xi^{-1}\} )
= \lvac \exp \biggl(\sum_{n=1}^\infty \frac{1}{n} (\frac{1}{\eta^n}- \frac{1}{\xi^n}) \alpha_n \biggr)
e^A \vac
\een
is equal to the inner product of $e^A\vac$ with
$\exp \biggl(\sum_{n=1}^\infty \frac{1}{n} (\frac{1}{\eta^n}- \frac{1}{\xi^n}) \alpha_{-n} \biggr)\vac$.
On the other hand one has an expansion:
\ben
\exp \biggl(\sum_{n=1}^\infty \frac{1}{n} (\frac{1}{\eta^n}- \frac{1}{\xi^n}) \alpha_{-n} \biggr)\vac
= \sum_\mu s_\mu(p_k= \frac{1}{\eta^k} - \frac{1}{\xi^k}, k \geq 1 ) \cdot |\mu \rangle,
\een
and
\ben
e^A\vac = \vac + \sum_{m, n \geq 0} A_{n,m}  \psi_{-m-1/2} \psi^*_{-n-1/2} \vac + \cdots,
\een
where $\cdots$ are higher order terms in $A_{m,n}$,
therefore by \eqref{eqn:Specialization-Tau},
one gets
\be
s_\mu(p_k= \frac{1}{\eta^k} - \frac{1}{\xi^k}, k \geq 1 ) = 0
\ee
except for $\mu = \emptyset$ or $(m|n)$ for some $m,n \geq 0$,
and
\ben
&& \sum_{m, n \geq 0}  s_{(m|n)}(p_k= \frac{1}{\eta^k} - \frac{1}{\xi^k}, k \geq 1 ) \cdot (-1)^n A_{n,m} \\
& = & (\xi-\eta) \sum_{m, n \geq 0} A_{m,n} \xi^{-m-1} \eta^{-n-1},
\een
and so
\be
 s_{(m|n)}(p_k= \frac{1}{\eta^k} - \frac{1}{\xi^k}, k \geq 1 )
 = (-1)^n  (\xi-\eta)  \xi^{-n-1} \eta^{-m-1}.
\ee
Therefore,
\be
\begin{split}
& \exp \biggl(\sum_{n=1}^\infty \frac{1}{n} (\frac{1}{\eta^n}- \frac{1}{\xi^n}) \alpha_{-n} \biggr)\vac \\
= & \vac + \sum_{m,n \geq 0} (\xi-\eta)  \xi^{-n-1} \eta^{-m-1} \psi_{-m-1/2} \psi^*_{-n-1/2} \vac.
\end{split}
\ee

Suppose now that $U$ is specified by an admissible basis
$\{f_n = z^{n+1/2} + \sum_{k < n} c_{n, k} z^{k+1/2}\}_{n\geq 0}$,
then by \eqref{eqn:Det-Admissible},
\be
\begin{split}
& \tau_U( \{\eta^{-1} \} - \{\xi^{-1}\} )  \\
= & 1 + \sum_{m, n \geq 0}  (\xi-\eta)  \xi^{-n-1} \eta^{-m-1} \\
& \cdot \begin{vmatrix}
c_{0,-m-1} & 1 \\
c_{1,-m-1} & c_{1,0} & 1 \\
c_{2,-m-1} & c_{2,0} & c_{2,1} & 1 \\
\vdots & \vdots & \vdots & \vdots \\
c_{n-1, m-1} & c_{n-1,0} & c_{n-1,1} & & \cdots & 1 \\
c_{n, m-1} & c_{n,0} & c_{n,1} & & \cdots & c_{n,n-1}
\end{vmatrix}.
\end{split}
\ee
As a corollary,
\be
a_{n,m} = \begin{vmatrix}
c_{0,-m-1} & 1 \\
c_{1,-m-1} & c_{1,0} & 1 \\
c_{2,-m-1} & c_{2,0} & c_{2,1} & 1 \\
\vdots & \vdots & \vdots & \vdots \\
c_{n-1, m-1} & c_{n-1,0} & c_{n-1,1} & & \cdots & 1 \\
c_{n, m-1} & c_{n,0} & c_{n,1} & & \cdots & c_{n,n-1}
\end{vmatrix}.
\ee
Of course this last identity can be obtained more directly by comparing the coefficients of
$\psi_{-m-1/2} \psi^*_{-n-1/2}$ on both sides of
\be
|U\rangle = e^A \vac.
\ee
The approach that we take in this Section relate $A(\xi, \eta)$ to $\tau_U(\{\eta^{-1} \} - \{\xi^{-1}\} )$
through some specialization of the Schur functions.
In Section \ref{sec:Okounkov} we will see an important application of this idea.

\section{Emergence of the Airy Function in Topological 2D Gravity}

In this Section we specialize the results in the preceding Sections to the case
of topological 2D gravity.

\subsection{Reductions}

Let $A$ be a subalgebra of $\bC((z^{-1}))$ such that
\be
A \cap \bC[[z^{-1}]] = \bC,
\ee
and define the $A$-reduced Sato Grassmannian by:
\be
\Gr^A_{(0)} := \{ U \in \Gr_{(0)} \;|\; f(s_1) U \subset U, \; \forall f \in A \}.
\ee
Then for $U \in \Gr^A_{(0)}$,
$\tau_U$ is a tau-function of the KP hierarchy,
such that the associated pseudo-differential operator $L$ satisfies
\be
\text{$f(L)$ is a differential operator, for all $f\in A$}.
\ee

The Airy curve $y=\half x^2$ defines a reduction to the KdV hierarchy by taking
$A= \bC[z^2]$.

\subsection{Noncommutative deformation of the Airy curve}

A consequence of Sato's approach to KP hierarchy is
the emergence of the noncommutative deformation theory of the Airy curve.
Quantize the Airy curve by understanding $p$ as $i \frac{\pd}{\pd x}$,
and rewrite the equation of the Airy curve as a differential operator $P_0 = \frac{\pd^2}{\pd x^2} + 2x$.
By a noncommutative deformation of $P_0$ we mean a differential operator of the form:
\be
P_\bt = \frac{\pd^2}{\pd x^2} + 2 u(x; \bt).
\ee
Witten Conjecture/Kontsevich Theorem states that
when $u(x; \bt)
= \frac{\pd^2 F(\bt)}{\pd t_0^2}$, $x= t_0$,
the operator $P_t$ satisfies the following evolution equations:
\be
\pd_{t_n} P_t = (2n+1)!!\cdot [(P_\bt^{(2n+1)/2})_+, P_\bt].
\ee

\subsection{Emergence of the Airy equation}

The wave-function for the operator $P_0= \pd_x^2 + 2x$ is a function $w(x; \xi)$ such that
\be
\pd_x^2 w + 2x \cdot w = \xi^2 \cdot w.
\ee
Perform a change of variables:
\be
\tilde{x} = 2^{-2/3} (\xi^2- 2x).
\ee
Then one gets:
\be
\pd_{\tilde{x}}^2 w = \tilde{x} \cdot w.
\ee

Recall the Airy functions $\Ai(x)$ and $\Bi(x)$ can be defined by the following integral representations:
\bea
&& \Ai(x) = \frac{1}{\pi} \int_0^\infty \cos( \frac{t^3}{3} + xt) dt, \\
&& \Bi(x) = \frac{1}{\pi} \int_0^\infty \big( \exp( - \frac{t^3}{3} + xt)
+ \sin (\frac{t^3}{3} + xt) \big) dt.
\eea
They are two solutions to the Airy equation
\be
\pd_x^2 y - x \cdot y  =0.
\ee
These functions and their first derivatives admit the following asymptotic expansions:
\bea
&& \Ai(x) \sim \frac{1}{2} \pi^{-1/2} x^{-1/4} e^{-\frac{2}{3} x^{3/2}}
\sum_{m=0}^\infty \frac{(6m-1)!!}{(6\sqrt{-2})^{2m} (2m)!} x^{-3m/2}, \\
&& \Ai'(x) \sim \frac{1}{2} \pi^{-1/2} x^{1/4} e^{-\frac{2}{3} x^{3/2}}
\sum_{m=0}^\infty \frac{(6m-1)!!}{(6\sqrt{-2})^{2m} (2m)!} \frac{6m+1}{6m-1} x^{-3m/2}, \\
&& \Bi(x) \sim   \pi^{-1/2} x^{-1/4} e^{\frac{2}{3} x^{3/2}}
\sum_{m=0}^\infty \frac{(6m-1)!!}{(6\sqrt{2})^{2m} (2m)!} x^{-3m/2}, \\
&& \Bi'(x) \sim  \pi^{-1/2} x^{1/4} e^{\frac{2}{3} x^{3/2}}
\sum_{m=0}^\infty \frac{(6m-1)!!}{(6\sqrt{2})^{2m} (2m)!} \frac{6m+1}{6m-1} x^{-3m/2}.
\eea

\subsection{Wave-function from the Airy function}
There exists a suitably chosen $c(\xi)= 2^{5/6}\pi^{1/2} \xi^{1/2} e^{\xi^3/3}$ such that
\ben
w(x; \xi) & = & c(\xi) \cdot \Ai(\tilde{x}) \\
& = & c(\xi) \cdot \frac{1}{2} \pi^{-1/2} \tilde{x}^{-1/4} e^{-\frac{2}{3} \tilde{x}^{3/2}}
\sum_{m=0}^\infty \frac{(6m-1)!!}{(6\sqrt{-2})^{2m} (2m)!} \tilde{x}^{-3m/2} \\
& = & c(\xi) \cdot \frac{1}{2} \pi^{-1/2} \cdot 2^{1/6} \cdot (\xi^2 - 2x)^{-1/4}
e^{-\frac{1}{3} (\xi^2-2x)^{3/2}} \\
&& \cdot \sum_{m=0}^\infty (-1)^m \frac{(6m-1)!! }{6^{2m} (2m)!}
(\xi^2-2x)^{-3m/2} \\
& = & e^{x\xi} \cdot  (1- 2\xi^{-2}x )^{-1/4}
e^{-\frac{1}{3} \xi^3 (1- 2\xi^{-2}x)^{3/2} - x\xi + \frac{1}{3} \xi^3} \\
&& \cdot \sum_{m=0}^\infty (-1)^m \frac{(6m-1)!! }{6^{2m} (2m)!} \xi^{-3m}
(1- 2\xi^{-2}x)^{-3m/2}
\een
is of the form
\be
w(x; \xi) = e^{x\xi} (1+ \frac{b_1(x)}{\xi} + \frac{b_2(x)}{\xi^2} + \cdots).
\ee
This is form of the  wave-function at  $t_k = 0$ for $k \geq 1$.
Now we have:
\bea
&&  w(0; \xi) = \sum_{m=0}^\infty (-1)^m \frac{(6m-1)!! }{6^{2m} (2m)!} \xi^{-3m}, \label{eqn:c} \\
&&  \pd_x w(0; \xi)
= \sum_{m=0}^\infty (-1)^{m+1} \frac{(6m-1)!!}{6^{2m} (2m)!}
\frac{6m+1}{6m-1} \xi^{-3m+1}. \label{eqn:q}
\eea
Denote these series by $c(\xi)$ and $q(\xi)$ respectively.
They are called the Faber-Zagier series \cite{Pand-Pixton-Zvonkine}.
They are closely related to the asymptotic series of $\Ai(z^2)$
and $\Ai'(z^2)$ respectively,
and then one has
\bea
&& D c(z) =  q(z), \\
&& D^2 c(z) = z^2 \cdot q(z),
\eea
where $D$ is the differential operator:
\be
D: = z+ \frac{1}{2z^2} - \frac{1}{z} \frac{\pd}{\pd z}.
\ee

\subsection{Airy functions in Kac-Schwarz's geometric characterization of $Z_{WK}$}

Consider the vector space $W$ spanned by $\{z^{1/2} D^n c(z) \}_{n \geq 0}$.
It is clear $W$ lies in the $\Gr_{(0)}$
and so $W$ determines a tau-function $\tau_W$ of the KP hierarchy.
Because it is clear that
\be \label{eqn:2-Reduced}
z^2 W \subset W,
\ee
this tau-function $\tau_W$ is $2$-reduced,
i.e.,
it is a tau-function of the KdV hierarchy.
It is also clear that
\be \label{eqn:D-Closed}
z^{1/2} D z^{-1/2} W \subset W.
\ee
Kac-Schwarz \cite{Kac-Schwarz} showed that this is equivalent to $\tau_W$
satisfying the puncture equation:
\be
L_{-1} \tau_W = 0.
\ee
Hence by Witten Conjecture/Kontsevich Theorem,
$\tau_W$ is the partition function of the topological 2D gravity.
Furthermore,
by combining \eqref{eqn:2-Reduced} and \eqref{eqn:D-Closed},
one gets
\be
z^{2n+1/2} D z^{-1/2} W \subset W, \;\;\; n \geq 0.
\ee
In {\em loc. cite.} it was shown these imply that $\tau_W$ satisfies
the Virasoro constraints:
\be
L_{n-1} \tau_W = 0, \;\; n \geq 0.
\ee
As pointed out by Looijenga \cite{Looijenga},
if one sets $\cL_n = -\half z^{2n+5/2} D z^{-1/2}$, $n \geq -1$,
then
\be
[\cL_m, \cL_n ] = (m-n) \cL_{m+n}.
\ee

\subsection{Airy function in Kontsevich's proof of Witten Conjecture}

Kontsevich \cite{Kontsevich} identified the generating series $Z_{WK}$ of
intersection numbers of $\psi$-classes on $\Mbar_{g,n}$ with $\tau_W$,
hence proved the Witten Conjecture.
(See \cite{Looijenga} for an exposition.)
He first established the Main Identity:
For $ g\geq 0$, $n \geq 1$ such that $2g-2+n > 0$,
\be \label{eqn:Main-Identity}
\begin{split}
& \sum_{\sum_{i=1}^n m_i = 3g-3+n} \corr{\tau_{m_1} \cdots \tau_{m_n} }_g
\prod_{i=1}^n \frac{(2m_i-1)!!}{\lambda_i^{2m_i+1}} \\
= & \sum_{\Gamma \in G_{g,n}} \frac{2^{-|V(\Gamma)|}}{|\Aut(\Gamma)|}
\prod_{e\in E(\Gamma)} \prod_{e \in E(\Gamma)} \frac{2}{\tilde{\lambda}(e)},
\end{split}
\ee
where $G_{g,n}$ is the set of isomorphism classes of connected trivalent
ribbon graphs of genus $g$ with $n$ boundary cycles labeled by $1, \dots, n$.
Given an edge $e$,
it is a band with two boundaries, labeled by $i$ and $j$, resepctively,
then
\be
\tilde{\lambda}(e) = \lambda_i + \lambda_j.
\ee
Next,
using the Main Identity and the Wick's Theorem,
Kontsevich obtain the following matrix model for the free energy.
Let
$$\Lambda = \diag(\lambda_1, \dots, \lambda_N)$$
be a positive definite hermitian $N \times N$-matrix,
and let
\be
t_i = t_i(\Lambda) = - (2i-1)!! \cdot \Tr \Lambda^{-(2i+1)}.
\ee
Then his Theorem 1.1 states the formal series $F(t_0(\Lambda), t_1(\Lambda), \dots)$
is an asymptotic expansion of $\log Z^{(n)}(\Lambda)$
when $\Lambda^{-1} \to 0$,
where
\be \label{eqn:Kontsevich-Model}
Z^{(N)}(\Lambda):= c_\Lambda \int \exp \Tr
\big( \frac{\sqrt{-1}}{6}X^3- \frac{1}{2} X^2\Lambda\big),
\ee
where the constant $c_\Lambda$ is chosen to be
$$c_\Lambda
= \det \big( \frac{1}{4\pi} (\Lambda \otimes 1 + 1 \otimes \Lambda) \big)^{1/2}$$
so that
$$c_\Lambda\int \exp \Tr \big(- \frac{1}{2} X^2\Lambda\big) dX
= 1.$$
Next,
Harish-Chandra formula was used to the reduce the matrix integral of Airy form
in Kontsevich model to a determinantal expression involves
the Airy function and its derivatives.
The relationship with Sato tau-function was established by \cite[Lemma 4.1]{Kontsevich} to get
\be
Z^{(N)}(\Lambda)
= \frac{\det\limits_{1 \leq i, j \leq N} (D^{j-1} a(\lambda_i^{-1}))}
{\det\limits_{1 \leq i, j \leq N} (\lambda_i^{1-j})}.
\ee
In the final step,
Kontsevich showed that when $N \to \infty$,
$\tau^{(N)}(\Lambda)$ gives $\tau_W$.
This gives an explicit construction of the Witten-Kontsevich tau-function in terms of the Airy function.

Let us note Airy functions $\Ai$ and $\Bi$ also appear in the setting of quantum spectral curves \cite{Zhou-Quantum}.

\subsection{Explicit fermionic expressions of Witten-Kontsevich tau-function}

The above formula for the Witten-Knotsevich tau-function was explained by Itzykson-Zuber \cite{Itzykson-Zuber}.
Inspired by \cite{Kac-Schwarz} and \cite{Kontsevich},
the author  \cite{Zhou-Explicit} found an explicit expression of $Z_{WK}$  as follows.
After suitable linear change of coordinates in $t_n$'s:
\be \label{eqn:t-T}
t_n = (2n+1)!! T_{2n+1} = \frac{(2n+1)!!}{2n+1} p_{2n+1},
\ee
and the boson-fermion correspondence,
the Witten-Kontsevich tau-function is a Bogoliubov transform in the fermionic picture:
\be
Z_{WK} = e^A \vac,  \quad \quad \quad
A =  \sum_{m,n \geq 0} A_{m,n} \psi_{-m-\frac{1}{2}}\psi_{-n-\frac{1}{2}}^*,
\ee
where the coefficients $A_{m,n} = 0$ if $m+n \not\equiv -1 \pmod{3}$ and
\ben
&& A_{3m-1,3n} = A_{3m-3, 3n+2}
=  \frac{(-1)^n }{36^{m+n}} \frac{(6m+1)!!}{(2(m+n))!} \\
&& \qquad\qquad \cdot \prod_{j=0}^{n-1} (m+j) \cdot \prod_{j=1}^{n} (2m+2j-1) \cdot (B_{n} (m) +\frac{b_n}{6m+1}), \\
&& A_{3m-2,3n+1}
=  \frac{(-1)^{n+1}}{36^{m+n}} \frac{(6m+1)!!}{(2(m+n))!} \\
&& \qquad\qquad  \cdot \prod_{j=0}^{n-1} (m+j) \cdot \prod_{j=1}^{n} (2m+2j-1) \cdot (B_{n}(m) +\frac{b_n}{6m-1}),
\een
where  $B_n(m)$ is a polynomial in $m$ of degree $n-1$ defined by:
\be
B_n(x) = \frac{1}{6} \sum_{j=1}^{n} 108^{j} b_{n-j} \cdot (x+n)_{[j-1]},
\ee
where
\be
(a)_{[j]} = \begin{cases}
1, & j = 0, \\
a(a-1) \cdots (a-j+1), & j > 0,
\end{cases}
\ee
and $b_n$ is a constant depending on $n$ defined by:
\be
b_{n} = \frac{2^n \cdot (6n+1)!!}{(2n)!}.
\ee
The following are some  examples of the coefficients $A_{m,n}$:
\ben
&& A_{3m-1,0} =  \frac{1}{36^m}  \frac{(6m+1)!!}{(2m)!} \cdot \frac{1}{6m+1}, \\
&& A_{3m-2,1} = -  \frac{1}{36^m} \frac{(6m+1)!!}{(2m)!} \cdot \frac{1}{6m-1}, \\
&& A_{3m-3,2} = \frac{1}{36^m} \frac{(6m+1)!!}{(2m)!} \cdot \frac{1}{6m+1}.
\een

\subsection{Transition matrix from admissible basis to normalized basis}

The above result was established by Virasoro constraints.
Balogh and Yang \cite{Balogh-Yang} rederived it from the point of view of Sato's construction of
the tau-function from admissible bases and normalized bases,
so it becomes a special case of the general results for KP hierarchy developed in this paper.
Let us briefly recall their beautiful results.
Suppose that $W \in \Gr_{(0)}$ is given by an admissible basis of the form
$\{z^{2n+1/2} a(z), z^{2n+3/2} b(z)\}_{n \geq 0}$,
where
\ben
a(z) = 1 + \sum_{n \geq 1} a_n z^{-n}, \;\;\;
b(z) = 1 + \sum_{n \geq 1} b_n z^{-n}.
\een
Define a matrix
$G(z) = \begin{pmatrix}
G(z)_{11} & G(z)_{12} \\ G(z)_{21} & G(z)_{22}
\end{pmatrix}
$
whose entries are given by
\begin{align*}
G(z)_{11} & = \sum_{n \geq 0} a_{2n} z^{-n}, & G(z)_{12} & = \sum_{n \geq 0} b_{2n+1} z^{-n}, \\
G(z)_{11} & = \sum_{n \geq 1} a_{2n-1} z^{-n}, & G(z)_{12} & = \sum_{n \geq 0} b_{2n} z^{-n}.
\end{align*}
Suppose that the corresponding normalized basis is given by
\be
z^n + \sum_{m \geq 0} A_{m,n} z^{-m-1}.
\ee
For $m, n \geq 0$, define a matrix $Z_{m,n}$ by
\be
Z_{m,n}: = \begin{pmatrix}
A_{2m+1,2n} & A_{2m+1, 2n+1} \\ A_{2m, 2n} & A_{2m, 2n+1}
\end{pmatrix}.
\ee
Then \cite[Theorem 1.1]{Balogh-Yang}:
\be \label{eqn:Balogh-Yang}
\sum_{m,n=0}^\infty Z_{m,n} x^{-m-1} y^{-n-1} = \frac{1}{x-y} (I -G(x) G(y)^{-1}).
\ee
When $a(z)$ and $b(z)$ are given by the following series  related to \eqref{eqn:c} and \eqref{eqn:q} respectively:
\bea
&&  a(z) = \sum_{m=0}^\infty  \frac{(6m-1)!! }{6^{2m} (2m)!} \xi^{-3m}, \label{eqn:a} \\
&&  b(z) = - \sum_{m=0}^\infty  \frac{(6m-1)!!}{6^{2m} (2m)!}
\frac{6m+1}{6m-1} z^{-3m+1}, \label{eqn:b}
\eea
$A_{m,n}$ are given by the explicit expressions in last subsection.
(This is when we use \eqref{eqn:t-T}.
If we use instead as in \cite{Balogh-Yang}
\be \label{eqn:-t-T}
t_n = -(2n+1)!! T_{2n+1} = \frac{(2n+1)!!}{2n+1} p_{2n+1},
\ee
then $a(z)$ and $b(z)$ should be replaces by $c(z)$ and $q(z)$ respectively. )
Furthermore,
the matrix $G$ is related to the $R$-matrix of in the theory of Frobenius manifold associated with
3-spin structures \cite{Witten-R-Spin}.
According to \cite{Pand-Pixton-Zvonkine},
when $a(z)$ and $b(z)$ are taken to be the Faber-Zagier series,
\be
R(z) = \begin{pmatrix}
\sum_{n \geq 0} b_{2n} z^{2n} & - \sum_{n \geq 0} b_{2n+1} z^{2n+1} \\
- \sum_{n \geq 0} a_{2n+1} z^{2n+1} & \sum_{n \geq 0} a_{2n} z^{2n}
\end{pmatrix}.
\ee
Then \cite[Theorem 1.2]{Balogh-Yang}:
\be
R(z) =z^{\frac{1}{6} \sigma_3} G(z^{-2/3}) z^{-\frac{1}{6} \sigma_3},
\ee
where $\sigma_3 = \begin{pmatrix} 1 & 0 \\ 0 & -1 \end{pmatrix}$.
See \cite{Balogh-Yang-Zhou} for generalizations of these results to all $r$-spin cases.

\subsection{Relationship with a formula of Okounkov}

\label{sec:Okounkov}

Let us now relate the formula \eqref{eqn:Balogh-Yang} of Balogh-Yang \cite{Balogh-Yang}
to a formula of Okounkov \cite{Okounkov}
used in a proof of the Witten Conjecture.
Since $\det G(z)  = 1$,
one has
\begin{align}
G(z) & = \begin{pmatrix}  \sum\limits_{j \geq 0} a_{2j} z^{-3j} & \sum\limits_{j \geq 0}  b_{2j+1} z^{-3j-1} \\
\sum\limits_{j \geq 0} a_{2j+1} z^{-3j-2} & \sum\limits_{j \geq 0}  b_{2j} z^{-3j}\end{pmatrix},  \\
G(z)^{-1} & = \begin{pmatrix}  \sum\limits_{j \geq 0} b_{2j} z^{-3j} & -\sum\limits_{j \geq 0}  b_{2j+1} z^{-3j-1} \\
-\sum\limits_{j \geq 0} a_{2j+1} z^{-3j-2} & \sum\limits_{j \geq 0}  a_{2j} z^{-3j}\end{pmatrix}
\end{align}
From
\be
\begin{split}
& \sum_{m,n=0}^\infty  x^{-m-1} y^{-n-1} \begin{pmatrix}
A_{2m+1,2n} & A_{2m+1, 2n+1} \\ A_{2m, 2n} & A_{2m, 2n+1}
\end{pmatrix} \\
= & \frac{1}{x-y} (I -G(x) G(y)^{-1}),
\end{split}
\ee
one gets four identities:
\ben
\sum_{m, n \geq 0} A_{2m+1,2n} x^{-m-1} y^{-n-1}
& = & \frac{1}{x-y} (1- \sum_{j\geq 0} a_{2j} x^{-3j} \cdot \sum_{k \geq 0} b_{2k} y^{-3k} \\
& + & \sum\limits_{j \geq 0}  b_{2j+1} x^{-3j-1} \cdot \sum\limits_{k \geq 0} a_{2k+1} y^{-3k-2} ), \\
\sum_{m, n \geq 0} A_{2m+1,2n+1} x^{-m-1} y^{-n-1}
& = & \frac{1}{x-y} ( \sum_{j\geq 0} a_{2j} x^{-3j} \cdot \sum_{k \geq 0} b_{2k+1} y^{-3k-1} \\
& - & \sum\limits_{j \geq 0}  b_{2j+1} x^{-3j-1} \cdot \sum\limits_{k \geq 0} a_{2k} y^{-3k} ), \\
\sum_{m, n \geq 0} A_{2m,2n} x^{-m-1} y^{-n-1}
& = & \frac{1}{x-y} (- \sum_{j\geq 0} a_{2j+1} x^{-3j-2} \cdot \sum_{k \geq 0} b_{2k} y^{-3k} \\
& + & \sum\limits_{j \geq 0}  b_{2j} x^{-3j} \cdot \sum\limits_{k \geq 0} a_{2k+1} y^{-3k-2} ), \\
\sum_{m, n \geq 0} A_{2m,2n+1} x^{-m-1} y^{-n-1}
& = & \frac{1}{x-y} (1 + \sum_{j\geq 0} a_{2j+1} x^{-3j-2} \cdot \sum_{k \geq 0} b_{2k+1} y^{-3k-1} \\
& - & \sum\limits_{j \geq 0}  b_{2j} x^{-3j} \cdot \sum\limits_{k \geq 0} a_{2k} y^{-3k} ).
\een
They can be combined into one beautiful identity:
\be \label{eqn:Balogh-Yang2}
\begin{split}
& \sum_{m, n \geq 0} A_{m,n} x^{-m-1} y^{-n-1} \\
= & \frac{1}{x-y} + \frac{1}{x^2-y^2} ( a(x) \cdot b(-y) - a(-y) \cdot b(x)).
\end{split}
\ee
To relate this to \cite{Okounkov},
we now use \eqref{eqn:Specialization-Tau} to get:
\be
\sum_{m, n \geq 0} a_{m,n} x^{-m-1} y^{-n-1} = \frac{1}{x-y} + \frac{1}{x-y} \tau_W(\{y^{-1}\}- \{x^{-1}\}).
\ee
Because $\tau_W(\bT)$ depends only on $T_{2n+1}$,
so we have
\ben
\tau_W(\{y^{-1}\}- \{x^{-1}\})
& = & \tau_W(\{y^{-1}\} + \{-x^{-1}\}) \\
& = & \lvac \exp \sum_{n=1}^\infty \frac{1}{n} (\frac{1}{y^n} + \frac{1}{(-x)^n}) \alpha_n) |W\rangle,
\een
so it is the inner product of $|W\rangle$
with
\ben
&&  \exp \sum_{n=1}^\infty \frac{1}{n} (\frac{1}{y^n} + \frac{1}{(-x)^n}) \alpha_{-n}) \vac
= \sum_\mu s_\mu(p_n = \frac{1}{y^n} + \frac{1}{(-x)^n}, n \geq 1) \cdot |\mu\rangle.
\een
The specialization of the Schur functions can be computed by the Jacobi-Trudy identities:
For $n \geq l(\mu)$,
\be \label{eqn:Jacobi-Trudy1}
s_\mu = \det (h_{\mu_i - i +j})_{1 \leq i, j \leq n},
\ee
and for $n \geq l(\mu^t)$,
\be \label{eqn:Jacobi-Trudy2}
s_\mu = \det (e_{\mu_i^t - i +j})_{1 \leq i, j \leq n}.
\ee
For the specialization with $p_n = y^{-n} + (-x)^{-n}$,
it is clear that the generating series for elementary symmetric functions $e_k$ and the complete
symmetric functions $h_k$ are given by
\be
E(t) = (1-x^{-1} t) (1+ y^{-1} t)= 1 + (y^{-1}-x^{-1}) t - x^{-1} y^{-1} t^2,
\ee
and
\be
H(t) = \frac{1}{(1 + x^{-1}t)(1-y^{-1}t)}
= \sum_{n \geq 0} (\sum_{i+j=n} x^{-i} y^{-j}) t^n
\ee
respectively.
It follows that
\ben
e_1 = y^{-1}-x^{-1}, \;\;\; e_2 = - x^{-1} y^{-1}, \;\;\; e_k = 0, \; k> 2,
\een
and
\be
h_n = \sum_{i+j=n} (-1)^i x^{-i} y^{-j} = \frac{(-x)^{-n-1}-y^{-n-1}}{(-x)^{-1} - y^{-1}}.
\ee
It follows from \eqref{eqn:Jacobi-Trudy2} that
when $l(\mu) > 2$,
\be
s_\mu(p_n = \frac{1}{y^n} + \frac{1}{(-x)^n}, n \geq 1) = 0.
\ee
Indeed,
since $\mu^t_1 = l(\mu) > 2$,
one has $e_{\mu^t_1-1+j} = 0$ for $j=1, \dots, n$,
i.e.,
the first row in the determinant on the right-hand side od \eqref{eqn:Jacobi-Trudy2} all vanish.
Therefore,
one only has to consider the case of $l(\mu) =1$ and $l(\mu) =2$,
for which we use \eqref{eqn:Jacobi-Trudy1} to get
\ben
&& s_{(m|0)} (p_k = \frac{1}{y^k} + \frac{1}{(-x)^k}, k \geq 1)
= h_{m+1} = \frac{(-x)^{-m-2}-y^{-m-2}}{(-x)^{-1} - y^{-1}}, \\
&& s_{(m|1)} (p_k = \frac{1}{y^k} + \frac{1}{(-x)^k}, k \geq 1)
= s_{(m+1,1)} =   \frac{(-x)^{-m-2} y^{-1} - (-x)^{-1} y^{-m-2}}{(-x)^{-1} - y^{-1}}, \\
&& s_{(m_1, m_2)|(1,0)}  (p_k = \frac{1}{y^k} + \frac{1}{(-x)^k}, k \geq 1) \\
& = & s_{(m_1+1, m_2+2)}  (p_k = \frac{1}{y^k} + \frac{1}{(-x)^k}, k \geq 1) \\
& = & \begin{vmatrix}
h_{m_1+1} & h_{m_1+2} \\
h_{m_2+1} & h_{m_2+2}
\end{vmatrix}
= \frac{(-x)^{-m_1-2} y^{-m_2-2} - (-x)^{-m_2-2} y^{-m_1-2}}{(-x)^{-1} - y^{-1}}.
\een
So we get
\ben
&&  \exp (\sum_{n=1}^\infty \frac{1}{n} (\frac{1}{y^n} + \frac{1}{(-x)^n}) \alpha_{-n}) \vac \\
& = & \vac + \sum_{m\geq 0} \frac{(-x)^{-m-2}-y^{-m-2}}{(-x)^{-1} - y^{-1}}
\cdot \psi_{-m-1/2} \psi^*_{-1/2} \vac \\
& - & \sum_{m\geq 0} \frac{(-x)^{-m-2} y^{-1} - (-x)^{-1} y^{-m-2}}{(-x)^{-1} - y^{-1}}
\cdot \psi_{-m-1/2} \psi_{-3/2}^*\vac \\
& - & \sum_{m_1 > m_2 \geq 0} \frac{(-x)^{-m_1-2} y^{-m_2-2} - (-x)^{-m_2-2} y^{-m_1-2}}{(-x)^{-1} - y^{-1}} \\
&&  \cdot \psi_{-m_1-1/2}\psi_{-3/2}^* \psi_{-m_2-1/2} \psi_{-1/2}^*\vac.
\een
Note we have
\ben
&& (z^{1/2} + \sum_{j \geq 1} a_j z^{-3j+1/2}) \wedge (z^{3/2} + \sum_{k \geq 1} b_k z^{-3k+3/2}) \\
& = & z^{1/2} \wedge z^{3/2} - \sum_{k \geq 1} b_k z^{-3k+3/2} \wedge z^{1/2}
+ \sum_{j \geq 1} a_j z^{-3j+1/2} \wedge z^{3/2} \\
& + & \sum_{j \geq k \geq 1} a_j b_k z^{-3j+1/2} \wedge z^{-3k+3/2}
- \sum_{1 \leq j < k} a_j b_k  z^{-3k+3/2} \wedge z^{-3j+1/2},
\een
so the relevant terms in $|W\rangle$ are
\ben
|W\rangle & = & \vac + \sum_{j \geq 1} a_j \cdot \psi_{-3j+1/2} \psi^*_{-1/2} \vac
+ \sum_{k \geq 1} b_k \cdot  \psi_{-3k+3/2} \psi^*_{-3/2} \vac \\
& + & \sum_{j \geq k \geq 1} a_j b_k \cdot \psi_{-3j+1/2} \psi_{-3k+3/2} \psi^*_{-3/2} \psi^*_{-1/2} \vac \\
& - & \sum_{1 \leq j < k} a_j b_k \cdot \psi_{-3k+3/2} \psi_{-3j+1/2}  \psi^*_{-3/2} \psi^*_{-1/2} \vac
+ \cdots,
\een
where $\cdots$ are irrelevant terms.
Taking inner product with $|W\rangle$ one then gets:
\ben
&& \tau_W(\{y^{-1}\}- \{x^{-1}\}) \\
& = & 1 +  \sum_{j \geq 1} a_j \frac{(-x)^{-3j-1}-y^{-3j-1}}{(-x)^{-1} - y^{-1}}
- \sum_{k \geq 1} b_k  \frac{(-x)^{-3k} y^{-1} - (-x)^{-1} y^{-3k}}{(-x)^{-1} - y^{-1}}  \\
& + & \sum_{j \geq k \geq 1} a_j b_k \cdot \frac{(-x)^{-3j-1} y^{-3k} - (-x)^{-3k} y^{-3j-1}}{(-x)^{-1} - y^{-1}} \\
& - & \sum_{1 \leq j < k} a_j b_k \cdot \frac{(-x)^{3k} y^{-3j-1} - (-x)^{-3j-1} y^{-3k}}{(-x)^{-1} - y^{-1}}\\
& = & \sum_{j,k \geq 0} a_j b_k  \cdot \frac{(-x)^{-3j-1} y^{-3k} - (-x)^{-3k} y^{-3j-1}}{(-x)^{-1} - y^{-1}} \\
& = &  \frac{a(-x) \cdot b(y) - a(y) b(-x)}{x + y}.
\een
Note we have obtained
\be
 \tau_W(\{x^{-1}\} + \{y^{-1}\})
= \sum_{j,k \geq 0} a_j b_k  \cdot \frac{x^{-3j-1} y^{-3k} - x^{-3k} y^{-3j-1}}{x^{-1} - y^{-1}},
\ee
which is a formula that has been used  in Okounkov \cite{Okounkov} to prove the Witten Conjecture.
The following are the first few terms:
\ben
&& \tau_W(\{x^{-1}\} + \{ y^{-1}\})
= \frac{1}{x^{-1}- y^{-1}}
\biggl((x^{-1} - y^{-1}) \\
& + & \biggl( \frac{5}{24} (x^{-4}-y^{-4})- \frac{7}{24} (x^{-1}y^{-3}-x^{-3}y^{-1}) \biggr) \\
& + & \biggl( \frac{385}{1152}(x^{-7}- y^{-7}) +\frac{5}{24} \cdot \frac{-7}{24} (x^{-4} y^{-3}-x^{-3}y^{-4}) \\
&& + \frac{-455}{1152} (x^{-1}y^{-6}-x^{-6}y^{-1}) \biggr) \\
& + & \biggl( \frac{85085}{82944} (x^{-10}-y^{-10})
+ \frac{385}{1152} \cdot \frac{-7}{24} (x^{-7}y^{-3}-x^{-3}y^{-7}) \\
&& + \frac{5}{24} \cdot \frac{-455}{1152}\cdot (x^{-4}y^{-6}-x^{-6}y^{-4})
 + \frac{-95095}{82944} (x^{-1}y^{-9}-x^{-9}y^{-1}) \biggr) \\
& + & \biggl( \frac{37182145}{7962624} (x^{-13} - y^{-13})
+ \frac{85085}{82944}\cdot \frac{-7}{24} (x^{-10}y^{-3}-x^{-3}y^{-10}) \\
&& + \frac{385}{1152} \cdot \frac{-455}{1152} (x^{-7}y^{-6}-x^{-6}y^{-7})
+ \frac{5}{24} \cdot \frac{-95095}{82944} (x^{-4} y^{-9}-x^{-9}y^{-4}) \\
&& + \frac{-40415375}{7962624} (x^{-1}y^{-12}-x^{-12}y^{-1}) \biggr) + \cdots.
\een

Now we get:
\ben
\sum_{m, n \geq 0} a_{m,n} x^{-m-1} y^{-n-1}
& = & -\frac{1}{x-y} + \frac{1}{x-y} \tau_W(\{y^{-1}\}- \{x^{-1}\}) \\
& = & - \frac{1}{x-y} +   \frac{a(-x) \cdot b(y) - a(y) b(-x)}{x^2 - y^2}.
\een
This is exactly \eqref{eqn:Balogh-Yang2}
if we interchange $x$ with $y$ and note
\be
a_{m,n} = A_{n,m}.
\ee

The following are the first few terms of $A(x,y)$:
\ben
A(x, y)
& = &  \frac{5}{24x y^3}-\frac{7}{24x^2y^2}+ \frac{5}{24 x^3 y}  \\
& + &  \frac{385}{1152xy^6}- \frac{455}{1152x^2y^5}+ \frac{385}{1152x^3y^4} \\
& - & \frac{385}{1152x^4y^3} + \frac{455}{1152x^5y^2}- \frac{385}{1152x^6y} \\
& + & \frac{85085}{82944xy^9}- \frac{95095}{82944x^2y^8}+ \frac{85085}{82944x^3y^7} \\
& - & \frac{43505}{41472x^4y^6} + \frac{45955}{41472x^5y^5} - \frac{43505}{41472x^6y^4} \\
& + & \frac{85085}{82944x^7y^3}- \frac{95095}{82944x^8y^2}+ \frac{85085}{82944x^9y} + \cdots.
\een

\subsection{Bosonic $N$-point functions of topological 2D gravity}

\label{sec:N-Point-2DGravity}

Let us conclude this paper by an application of the explicit expression of the Witten-Kontsevich tau-function
in \cite{Zhou-Explicit}
and its geometric interpretation given by \cite{Balogh-Yang}.
We combine these results with  Theorem \ref{thm:Bosonic-N-Point} to get:

\begin{thm}
For $F= \log Z_{WK}$,
\be \label{eqn:WK-N-Point}
\begin{split}
& \sum_{j_1,\dots, j_n \geq 1}
\frac{\pd^n F}{\pd T_{j_1} \cdots \pd T_{j_n} } \biggl|_{\bT =0}
  \xi_1^{-j_1-1}\cdots \xi_n^{-j_n-1} \\
= & (-1)^{n-1} \sum_{\text{$n$-cycles}}  \prod_{i=1}^n \hat{A}(\xi_{\sigma(i)}, \xi_{\sigma(i+1)}),
\end{split}
\ee
where
\be
\hat{A}(\xi_i, \xi_j) = \begin{cases}
A(\xi_i, \xi_i), & \text{if $i=j$}, \\
\frac{1}{\xi_i - \xi_j} + A(\xi_i, \xi_j), & \text{if $i \neq j$}.
\end{cases}
\ee
\end{thm}

For a generalization of this result to Witten's r-spin curves,
see \cite{Balogh-Yang-Zhou}.

Let us look at some special cases of \eqref{eqn:WK-N-Point}.
For $n=1$,
the prediction of \eqref{eqn:WK-N-Point} is
\be 
 \sum_{j}
\frac{\pd F}{\pd T_{j} } \biggl|_{\bT =0} \xi^{-j-1} 
= A(\xi, \xi).
\ee
By  \eqref{eqn:Balogh-Yang2},
\ben
A(\xi, \xi) &=&  \lim_{x \to \xi}
\biggl( \frac{1}{x-\xi} + \frac{1}{x^2-\xi^2} ( a(x) \cdot b(-\xi) - a(-\xi) \cdot b(x)) \biggr) \\
& = & \lim_{x \to \xi} \frac{1}{2x} (1+ a'(x) \cdot b(-\xi) - a(-\xi) \cdot b'(x)) \\
& = & \frac{1}{2\xi} (1+ a'(\xi) \cdot b(-\xi) - a(-\xi) \cdot b'(\xi)).
\een
By \eqref{eqn:One-Point-All-Genera2} one has
\be
\frac{\pd^n F}{\pd T_{j} } \biggl|_{\bT =0} \xi^{-j-1}
= \sum_{g\geq 1}  \frac{(6g-3)!! }{24^gg! \xi^{6g+1}}.
\ee
So one gets an identity:
\be
a'(\xi) \cdot b(-\xi) - a(-\xi) \cdot b'(\xi)
= -1 + 2 \sum_{g\geq 1}  \frac{(6g-3)!! }{24^gg! \xi^{6g}}.
\ee
It will be interesting to prove this directly by combinatorial method.
For $n=2$,
the prediction of \eqref{eqn:WK-N-Point} is
\be
 \sum_{j,k}
\frac{\pd^2 F}{\pd T_{j} \pd T_k} \biggl|_{\bT =0} \xi_1^{-j-1} \xi_2^{-k-1}
= - \hat{A}(\xi_1, \xi_2) \hat{A}(\xi_2, \xi_1).
\ee
It is interesting to compare this with Dijkgraaf's formula for two-point function \cite{Liu-Xu}.

 \end{document}